\documentclass[letter,11pt]{article}
\pdfoutput=1

\usepackage{amsmath,amssymb, amsthm}
\usepackage{complexity}
\usepackage{float}
\usepackage{algorithm}
\floatname{algorithm}{Protocol}
\usepackage{multirow}
\usepackage{graphicx}
\usepackage[margin=1in]{geometry}
\usepackage[normalem]{ulem}

\newtheorem{theorem}{Theorem}
\newtheorem{lemma}{Lemma}
\newtheorem{corollary}{Corollary}

\newtheorem{definition}{Definition}

\newcommand{\graphstate} {brickwork state }
\newcounter{linenum}

\newcommand{\ket}[1]{\left\lvert #1 \right\rangle}
\newcommand{\bra}[1]{\left\langle #1 \right\lvert}

\newcommand{\AR}[2][c]{$$\begin{array}[#1]{lllllllllllllll}#2\end{array}$$}

\def\EQ#1{\begin{eqnarray}#1\end{eqnarray}}

\usepackage[affil-it]{authblk}
\begin{document}
\allowdisplaybreaks[3]

\title{Unconditionally Verifiable Blind Quantum Computation}

\author[1,2]{Joseph F. Fitzsimons}
\author[3,4]{Elham Kashefi}
\affil[1]{Singapore University of Technology and Design,}
\affil[ ]{8 Somapah Road, Singapore 487372}
\affil[2]{Centre for Quantum Technologies, National University of Singapore,}
\affil[ ]{3 Science Drive 2, Singapore 117543}
\affil[3]{School of Informatics, University of Edinburgh,}
\affil[ ]{10 Crichton Street, Edinburgh EH8 9AB, UK}
\affil[4]{CNRS - Télécom ParisTech}
\affil[ ]{23 Avenue d'Italie, Paris 75013, France}
\maketitle

\begin{abstract}
Blind Quantum Computing (BQC) allows a client to have a server carry out a quantum computation for them such that the client's input, output and computation remain private. A desirable property for any BQC protocol is verification, whereby the client can verify with high probability whether the server has followed the instructions of the protocol, or if there has been some deviation resulting in a corrupted output state. A verifiable BQC protocol can be viewed as an interactive proof system leading to consequences for complexity theory. 
The authors, together with Broadbent, previously proposed a universal and unconditionally secure BQC scheme where the client only needs to be able to prepare single qubits in separable states randomly chosen from a finite set and send them to the server, who has the balance of the required quantum computational resources. In this paper we extend that protocol with new functionality allowing blind computational basis measurements, which we use to construct a new verifiable BQC protocol based on a new class of resource states. We rigorously prove that the probability of failing to detect an incorrect output is exponentially small in a security parameter, while resource overhead remains polynomial in this parameter. The new resource state allows entangling gates to be performed between arbitrary pairs of logical qubits with only constant overhead. This is a significant improvement on the original scheme, which required that all computations to be performed must first be put into a nearest neighbour form, incurring linear overhead in the number of qubits. Such an improvement has important consequences for efficiency and fault-tolerance thresholds.
\end{abstract}
 \newpage
\section{Introduction}

Scalable quantum computing has proven extremely difficult to achieve, and when the technology to build large scale quantum computers does become available it is likely that they will appear initially in small numbers at a handful of centers. How will a user interface securely with such a quantum computer?  A solution to this problem is offered by blind quantum computing (BQC), which enables a classical client (Alice) with limited quantum technology to delegate a computation to the quantum server(s) (Bob) in such a way that the privacy of the computation is preserved \cite{Childs,AS06,BFK09,Dorit,RUV13,BGS13}.

Blind classical computing (the notion of ``computing with encrypted data") was proposed by Feigenbaum~\cite{Feigen86} and then extended 
by Abadi, Feigenbaum and Killian in 
a client server setting~\cite{AFK}. They showed that a randomized classical polynomial time client can encrypt and delegate general instances of 
certain problems in  $\NP$\footnote{A problem is in the class $\NP$ if one can verify its answers efficiently; it is $\NP$-hard if it is as hard as any problem in $\NP$.} to a powerful but untrusted server. Remarkably, they also proved that the decision of no $\NP$-hard function can be encrypted in this way if unconditional security is required,\footnote{A crypto system is unconditionally (computationally) secure if it is secure even when the adversary has unlimited (restricted) computing power.} unless the polynomial hierarchy collapses at the third level. The idea of computing known circuits on encrypted data, while requiring the encryption and decryption procedures be independent of the complexity of the function to be evaluated, was introduced earlier by Rivest, Adleman and Dertouzous in a scenario restricted to computational security~\cite{RAD78} shortly after the invention of RSA~\cite{RSA}. The problem of creating such a scheme, known as fully homomorphic encryption, remained open for 30 years before being settled by Gentry in 2009 \cite{Gentry09}, leading to one of the most active areas of research in modern cryptography \cite{Vaik12} \footnote{While several attempts have been made in recent years to find homomorphic encryption schemes which allow for the evaluation of certain quantum operations \cite{rohde2012quantum,tan2014quantum,broadbent2014quantum,ouyang2015quantum}, a quantum analogue of fully homomorphic encryption remains elusive \cite{yu2014limitations}.}.

The first example of blind quantum computation was proposed by Childs~\cite{Childs} based on the idea of encrypting input qubits with a quantum one-time pad~\cite{AMTW00,BR00}. At each step, the client sends the encrypted qubits to the server, which applies a known quantum gate. Finally, the server returns the quantum state for the client to decrypt with their key. Cycling through a fixed set of universal gates ensures that the server learns nothing about the circuit. The next quantum blind protocol with the possibility of detecting a cheating server was proposed by Arrighi and Salvail~\cite{AS06}. In their scheme, the client gives the server multiple quantum inputs, most of which are \emph{decoys} (not intended to be part of the desired computation), but rather are used to detect the server's deviation. This leads to a trade-off on the server side between gaining information and not disturbing the system, and achieves cheat-sensitive security against individual attacks for a set of classical functions called \emph{random verifiable}, where it is possible for the client to efficiently generate random input-output pairs. Extending these results, together with Broadbent, we presented the first universal blind quantum computing (UBQC) protocol~\cite{BFK09} in the measurement-based model \cite{RB01,Mcal07}, where the only requirement for the client is a classical computing machine and a very weak quantum instrument, a random single qubit generator, a currently available technology as we have demonstrated recently \cite{BKBFZW11}. Aside from the cryptographic scenario, a scheme based on a quantum authentication protocol\footnote{The parties aim to communicate messages over an untrusted channel in such a way that the receiver can authenticate the sender.} was proposed by Aharonov, Ben-Or and Eban~\cite{Dorit}, showing that any language in $\BQP$ has an interactive proof system with a verifier accessing a constant-size quantum computer. This work was complemented by a recent groundbreaking result of Reichardt, Unger and Vazirani on the command of quantum systems via rigidity of CHSH games \cite{RUV13}, leading to further work on device independent verifiable blind quantum computing \cite{GKW15,HPF15}. 

Recent years have seen an explosion of interest in the topic of blind quantum computing. This includes, for example, the extension of measurement-based UBQC to various setting \cite{Morimae12,SKM13,MDK15,Morimae12,MF12b,LCWW14}, addressing key questions regarding the effect of the noise \cite{FM12,CVK13}, the creation of new protocols to optimize communications requirements \cite{GMMR13,MPF13,PF14}, the development of privacy amplification techniques, similar to those applicable to quantum key distribution, to combat the adverse effect of imperfect devices on blindness \cite{DKL11}, experimental demonstrations \cite {BKBFZW11,FBS14,BFKW13}, and new cryptographic applications \cite{Q-coins,BGS13}.

A desirable property for any UBQC protocol is verifiability, whereby the client has a mechanism to verify the correctness of a delegated computation. The motivation for this stems from the broad range of computations which can be performed on a quantum computer. For problems which are in $\NP$, the solution can be efficiently verified, at least in principle, using a witness. However, for other problems which can be efficiently computed using quantum computation, such as quantum simulation~\cite{GAN14}, a dishonest server cannot be detected in such a way. The ability to compute with encrypted data, while hiding the underlying function, has opened up new approaches to the problem of verification~\cite{BFK09,Dorit,RUV13}. The main contributions of the present paper are to make rigorous the foundations of measurement-based UBQC and to present a new verification protocol which we prove to be secure against the most general adversarial behavior of the server. Using this protocol, the client can verify with high probability whether Bob has followed the instructions of the protocol and the output state is indeed in the correct form, or if there has been a deviation resulting in an incorrect output state. The central idea is based on the insertion of randomly prepared single qubits (called \emph{traps}), blindly isolated from the actual computation, which can act as such a witness. Here, even the computation of the test (measurement of the qubits) can be performed blindly by an untrusted server as we have demonstrated recently \cite{BFKW13}.

The verification scheme we present here makes use of similar elements as suggested in \cite{BFK09}: trap computations are used to detect errors, and a fault-tolerant encoding of the computation is used to amplify the detection rate. While the proof sketch for the effectiveness of verification in the original UBQC paper did not consider the most general adversary, we prove that the modified scheme we present here detects or corrects any possible deviation by the server, except with probability which is exponentially suppressed. In order to do so we introduce new universal resource states beyond the original brickwork state introduced in \cite{BFK09}. The first such family is a simple modification of the brickwork state which allows for the embedding of an arbitrary trap qubit, which leads to an inverse polynomial probability of detecting a deviation from the computation. In order to achieve a higher rate of detection, we introduce a second resource state which overcomes the locality limitations inherant in the brickwork state. This allows for the inclusion of a polynomial number of trap qubits and fault-tolerant implementation of the target computation based on the topological scheme of Raussendorf, Harrington and Goyal \cite{RHG07}. Together, these two new features allow for the probability of failing to detect or correct a deviation from the protocol to be made exponentially small. In this work we deal only with the stand-alone security definitions, as composable secuity follows from recent follow-up work by Dunjko \textit{et al} \cite{DFPR13}.

The remainder of the paper is organized as follows. Section 2 and 3 summarize various required concepts from measurement-based quantum computing and also the original UBQC scheme presented in \cite{BFK09}.  In order to construct our new verifiable UBQC protocol we first introduce the concept of dummy qubits in Section 4, where we assume Alice now can prepare a qubit randomly chosen not only in the equatorial plain, as in the original UBQC scheme, but also from the set $\{\ket 0, \ket 1\}$. The latter qubits are called dummy qubits as they have no effect on the actual underlying computation. However, they permit the blind construction of isolated trap qubits in the state $\ket {+_\theta}$ as explained in Section 6 where the core concept of verification is introduced. In order to deal with both universality and verification, in Section 5 we introduce two new resource state called the cylinder brickwork and dotted-complete graph states. The use of this scheme is expected to lead to substantially increased thresholds for fault tolerant computing in the blind setting. A threshold for fault-tolerant blind computation in the absence of verification based on this fault-tolerance scheme was previously calculated as $4.3 \times 10^{-3}$ by Morimae and Fujii \cite{FM12}. As shown in Section 6, introduction of a single blind isolated trap qubit leads to a verifiable blind quantum computing protocol with security polynomial in the total number of qubits. In order to boost the security while maintaining universality a new scheme has to be constructed. This is done in Section 7 where we put together various constructions of the previous sections to present the main result of this paper, a universal exponentially-secure verifiable blind quantum computing protocol. 
\section{Preliminaries}

Measurement-based quantum computing (MBQC) \cite{RB01,Mcal07} is a novel form
of quantum information processing, where the key twin notions that
distinguish quantum information processing from its classical
counterpart, entanglement (creating non-local correlations
between quantum elements) and measurement (observing a quantum
system), are the explicit driving force of computation. More
precisely, a measurement-based computation consists of a phase in
which a collection of qubits are set up in a standard entangled state.
Measurements are then made on individual qubits and the outcomes of
the measurements may be used to determine further adaptive
measurements. Finally, again depending on measurement outcomes, local
adaptive unitary operators, called corrections, are applied to some
qubits; this allows the elimination of the indeterminacy introduced by
measurements. Conceptually MBQC separates the quantum and classical
aspects of computation; thus it clarifies, in particular, the
interplay between classical control and the quantum evolution process.
The UBQC protocol explores this unique feature of MBQC as it has been
proven to be conceptually enlightening to reason about distributed computing tasks using this approach \cite{MS08}. We begin by describing all the
required elements for an MBQC protocol and then move to the particular
family of distributed MBQC protocols for hiding various aspects of
a given computation.

\subsection{Single party (undistributed) MBQC protocol}

A formal language to describe in a compact way the operations needed for the MBQC model was proposed in \cite{Mcal07}. In this framework every MBQC algorithm (usually referred to as an MBQC pattern) involves a sequence of operations such as entangling gates, measurements and feed-forwarding of outcome results to determine further measurement bases. A measurement pattern, or simply a pattern, is defined by a choice of a set of working qubits ($V$), a subset of input qubits ($I$), another subset of output qubits ($O$), and a finite 
sequence of commands acting on qubits in $V$. Therefore, we consider patterns associated with the so-called open graphs.
\begin{definition}
	\label{def_opengraph}
	An \emph{open graph} is a triplet $(G, I, O)$, where $G = (V, E)$ is a undirected graph, and $I, O \subseteq V$ are respectively called input and output vertices.
\end{definition}

\noindent Following the terminology of \cite{Mcal07}, a single party MBQC protocol consists of three elements:

\begin{enumerate}
\item A uniform family of open graph states $\{(G_{n,m},I_n,O_n)\}_n$ over $m$ vertices associated with individual
qubits, where $n$ is the size of the input/output space of the
underlying computation. In this paper we deal only with those MBQC protocol that implements a unitary operator over their input space and hence the size of the output space is the same as the input space, but this is not a restriction and we can extend this treatment to any general completely positive trace preserving map by padding the input and output spaces. Further, for simplicity, we will assume that the input is always a pure state, though again this treatment can be extended to the general case. We usually assume that $|I|=|O|=n$,
however sometime $n$ is taken to be strictly larger than the dimension
of the input/output Hilbert space due to the existence of auxiliary
input or output qubits (as in later protocols which incorporate trap
qubits). In order to have uniform notation, for the latter case, we
will still use $I$/$O$ to be the class of all
non-prepared/non-measured qubits where it is strictly larger than the
class of all input/output qubits. By the term ``uniform family'' we
simply mean that for any protocol there exist a classical Turing
machine that for a given input of the size $n$ describes the required
graph over $m\ge n$ vertices. If the underlying geometry of the graph
is regular, for example being one-dimensional lines, two-dimensional
regular lattices or brickwork graphs (as we describe later), then
instead of referring to the Turing machine to define the uniform
family we simply use fixed parameters such as the size of the line or
lattice to specify the graphs. For any fixed input size $n$ the graph
$G_{n,m}$ describes the initial quantum state of the protocol. Given
an arbitrary state of the input qubits corresponding to the input
vertices of the graph, one prepares $m-n$ qubits in the state $\ket
+=\frac 1 {\sqrt 2} (\ket 0 + \ket 1)$ corresponding to all non-input
qubits ($I^c$) in the graph and then apply \textsc{ctrl}-$Z$ operator
between qubits $i$ and $j$, if the corresponding vertices in $G_{n,m}$
are connected. Note that since the \textsc{ctrl}-$Z$ gate is symmetric the
direction of the edge is not important and hence we are working with
undirected graphs. We will usually refer to the obtained quantum state
based on the graph $G_{n,m}$ as the graph state $G_{n,m}$, unless a
different notation is more appropriate, also for simplicity we drop
the indices.\\

\item A set of angles $\phi_i \in A$ where $A\subseteq [0,\ 2\pi)$ for all non-output qubits, to describe a collection of single qubit $(X,Y)$-measurements, that is measurement in the bases $\frac 1 {\sqrt 2}(\ket 0 \pm e^{i\phi_i} \ket 1)$. For the specific class of MBQC protocols that we discuss in this paper we require the angles to specify a collection of measurement bases, such that individual measurements are unbiased with respect to the initial state. This is an essential ingredient for the blindness property that we define later. Without loss of generality, we can fix the set from which the angles are chosen to be $A=\{0, \pi/4, 2\pi/4, \cdots , 7\pi/4 \}$. We will discuss later how this combination of angles and particular families of graph states leads to approximate universality. \\

\item The last ingredient is the structure of the dependency among the measurements. It is known that despite the probabilistic nature of the measurements, an MBQC protocol can implement a unitary computation over the input space by introducing a casual structure over the measurements. This is done by allowing any measurement on qubit $i$ to be dependent on the result of some (possibly none) previously measured qubits. Let $s_i \in \{0,1\}$ be the classical result of the measurement at qubit $i$. There are two type of dependencies, called $X$ and $Z$ dependency. If a measurement at qubit $i$ is $X$ or $Z$ dependent on the $s_j$ where qubit $j$ has been already measured then the actual angle of the measurement of qubit $i$ during the protocol run is $(-1)^{s_j}\phi_i$ or $\phi_i+s_j\pi$ respectively. Naturally one needs a non-cyclic structure to be able to run such dependencies and for an arbitrary graph such construction (if it exists) is formalized by the notion of the \emph{flow} of the graph \cite{flow06,gflow07}. Intuitively, flow captures the propagation of quantum information as the resource state is measured, identifying the locations where measurement-dependent corrections should be made (see Figure~\ref{flow}). A flow is defined by a function ($f:O^c \rightarrow I^c$) from the measured qubits to non-input qubits and a partial order ($\preceq$) over the vertices of the graph such that $\forall i:\; i \preceq f(i)$ and $\forall j \in N_G(i):\; f(i) \preceq j$, where $N_G(i)$ denotes the neighborhood of vertex $i$ in $G$. Each qubit $k$ is $X$ dependent on $f^{-1}(k)$ and $Z$ dependent on all qubits $l$ such that $k \in N_G(f(l))$. Note that if the dependency set is empty, that is there is no qubit $q$ such that $q=f^{-1}(k)$ or $q\in N_G(f(l))$ then we set the convention that the corresponding value of $s_q$ is zero and hence we can use the same formulas ($(-1)^{s_j}\phi_i$ or $\phi_i+s_j\pi$) to compute the dependent angles. For a given graph, once the input and output qubits have been labeled, the flow, if it exists, is uniquely determined.

\end{enumerate}

\begin{figure} [h]
\begin{center}
\includegraphics{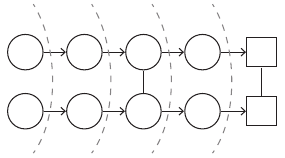}
\caption{An open graph state with flow. The boxed vertices are the output (non-measured) qubits and the circular vertices are the measured qubits. The flow function is represented as arrows (representing the $X$ dependency between measured qubits) and the partial order on the vertices (measurement order) is given by the dotted partition sets. One can see easily how the flow highlights the underlying circuit implemented by the measurement pattern.}
\label{flow}
\end{center}
\end{figure}

The above describes only a non-distributed (single party) MBQC protocol, that is a protocol where a party both prepares the graph state and performs the sequence of the dependent measurements according to the order given by the flow (see \cite{RB01,Mcal07} for more details on MBQC computation). One can easily extend the above definition to the distributed setting where different elements of the protocol are accessible and known only to specific parties and through classical/quantum communication the parties collaborate to perform a specific computation. Consider a simple two-party example where Alice has the information about the angles and Bob has the information about the graph and hence he can calculate the flow. Then they can collaborate to perform the corresponding computation as follows: first Bob prepares the required graph state and asks Alice to send him the classical information about the angles of the measurement, Bob then computes the dependency and performs the measurement and so forth. The purpose of this paper is to describe a family of such distributed protocols where, despite the communication, Alice can keep the measurement angles hidden from Bob. We then show that, for certain carefully chosen graph families, hiding these angles is sufficient to hide the full underlying computation together with the input and outputs.

\subsection{2-Party (distributed) Hiding Protocols}

We define a specific family of two-party (Alice and Bob) MBQC protocols (which we term hiding protocols) that can be shown to be ``blind'' in the sense that Alice can hide information from Bob. For simplicity, instead of working with a family of graphs representing the computation over an arbitrary size input, we fix the input size to be $n$ and we denote by $m \geq n$ the total number of vertices in the graph and hence the total number of qubits in the equivalent single-party protocol. Note that if we desire to have an efficient protocol, then we restrict the computation of the protocol to be of the polynomial size by requiring that $m=\text{Poly}(n)$. However blindness is independent of any complexity assumptions so we do not, in general, restrict the size of $m$.

The protocol will be interactive having $m-n$ steps if the output is quantum or $m$ steps if the output is classical, where at each step a single qubit is measured. In practice we can parallelize the protocol to $D$ steps, where $D$ is the depth of the partial order of the flow of the graph \cite{BK09,BKP10}. This is due to the special structure of the partial order of the qubits defined by the flow function whereby all the qubit in the same class of the partial order are independent of each other and hence can be measured in parallel, i.e. at the same time. However this parallelization will make no difference to the concept of blindness that we are concerned with, so we keep the simple convention that at each step only one qubit is measured.  Furthermore we assume for the case of classical output that all of the output qubits are measured in the final step with a Pauli $X$ measurement. Again this is simply a convention for the discussion in our paper and in general the output qubits could be measured with any angles and in different steps depending on the flow construction. Such a convention does not affect universality, as the circuit being implemented can simply be modified to replace measurements in arbitrary bases with measurements in fixed bases preceded by an appropriate local rotation. 

We will denote by $\mathbf{s}$ a sequence of length $m-n$ with value in $\{0,1\}$ describing the result of the non-output measurements performed so far. In the case of classical output, where output qubits are measured as the last $n$ steps, $\mathbf{s}$ is a sequence of length $m$. The value associated with a qubit that is not yet measured are set to $0$, and hence at the beginning of the protocol before any measurement being performed we set $\mathbf{s}=0,0,\cdots,0$. We will denote by $\mathbf{s}_{\leq i}$ the prefix of length $i$ of $\mathbf{s}$ and elements of $\mathbf{s}$ are denoted by $s_i$. Whenever adding the values of $s_i$ and $s_j$ we define their sum modulo $2$. All the qubits in the protocol are enumerated in such a way that at position $i$ all qubits with label less than $i$ are measured before measuring qubit $i$. Any total ordering of the qubits consistent with partial ordering of the flow will work and as a result the measurement at qubit $i$ will depend only on the string $\mathbf{s}_{< i}$\,.

We describe first a generic hiding protocol with quantum input and output (Protocol \ref{prot:GBQCQ}) and one with classical input and output (Protocol \ref{prot:GBQCC}) and then formalize various derivatives of them to obtain universal, blind and verifiable protocols. Protocol \ref{prot:GBQCC} is exactly the same as Protocol \ref{prot:GBQCQ} except that the steps for encoding input are removed and all the output qubits are measured in the Pauli $X$ basis. We retain the common text between the protocols so that they can be understood individually. Note that the reason we chose the measurement of the output qubits to be in the Pauli $X$ basis is purely for simplicity of presentation, so that the same evaluation function $C$ of the non-output measurements, in Protocol \ref{prot:GBQCQ}, can be used for the output qubits. However one could add separate evaluation function for the output qubit measurement to perform Pauli $Z$ measurement over them.

The outline of the main protocol is as follows. Alice has in mind a unitary operator $U$ that is implemented with a measurement pattern on some graph state $G$ with its unique flow function $f$, and measurements angles in $A=\{0, \pi/4, 2\pi/4, \cdots , 7\pi/4 \}$. This pattern could have been designed either directly within the MBQC framework or via translation from a circuit construction. The pattern assigns a measurement angle~$\phi_i$ to each qubit in $G$, however during the execution of the pattern, the actual measurement angle $\phi'_i$ is a modification of $\phi_i$ that depends on previous measurement outcomes instructed by $f$ in the following way \cite{flow06,gflow07}: 
\begin{align*}
\phi'_i = (-1)^{s_{f^{-1}(i)}} \phi_i + \sum_{j: \, i\in N_G(f(j))} s_j \pi \, .
\end{align*}
As said before, in a standard MBQC pattern all the non-input qubits are prepared in the state $\ket +$ and all the input qubits in the desired input state $\ket I$. Considering such quantum input allows for the possibility of Alice having additional capabilities allowing her to produce arbitrary input states, or for the possibility that the input state is supplied on Alice's behalf by a third party.

In our protocols, in order to hide the information about the angles some randomness has to be added to the preparation and consequently the measurements have to be adjusted to compensate for this initial randomness to obtain the correct outcome. This randomisation has three components:

\begin{itemize}
\item A set of random angles $\theta$ used to hide the true measurement angles $\phi$, 
\item A set of random bits $r$ used to hide measurement outcomes, 
\item A set of random bits $x$ used, along with $\theta$, to hide any quantum input via a one-time pad. 
\end{itemize}
 
\noindent Alice prepares all the non-input qubits in $\ket {+_{\theta_i}}=\frac 1 {\sqrt 2} (\ket 0 + e^{i\theta_i} \ket 1)$ for some randomly chosen $\theta_i \in A$ and also applies a modified version of a full quantum one-time pad encryption over the input qubits using random keys $x_i \in \{0,1\}$ and $\theta_i \in A$ in the following way:
\begin{align*}
\ket e = X_1^{x_1} Z_1(\theta_1) \otimes \ldots \otimes  X_n^{x_n} Z_n(\theta_n) \ket I \, ,
\end{align*}
before sending all qubits to Bob. After that, Bob entangles
qubits according to $G$. Note that this unavoidably
reveals upper bounds on the dimensions of the underlying quantum computation, corresponding to the length of the input and depth of the
computation. The computation stage
involves interaction: for each qubit, Alice sends Bob a classical message $\delta_i \in A$ to tell him in
which basis (in the~$(X,Y)$ plane) he should measure the qubit. This angle is computed in such a way as to correct for the one-time padding of the input qubits and the random rotation of the non-input qubits, as follows:
\begin{align*}
\delta_i = (-1)^{x_i+s_{f^{-1}(i)}} \phi_i + \sum_{j: \, i\in N_G(f(j))} s_j \pi + \theta_i + r_i \pi \, ,
\end{align*}
where the last term $r_i\pi$, with a randomly chosen $r_i \in \{0,1\}$, is added to hide the correct classical outcome of the measurement from Bob without affecting the overall computation (see correctness proof below). Bob then performs the measurement and  communicates the outcome $b_i$ to Alice. Alice's choice of angles in future rounds will depend on these values, hence she will correct the obtained outcome by setting $s_i:=b_i \oplus r_i$. If Alice is computing a classical function, the protocol finishes when all  qubits are measured (Protocol \ref{prot:GBQCC}), as the classical outputs are encoded in the measurement outcomes sent to Alice. If she is computing a quantum function, Bob returns to her
the final qubits (Protocol \ref{prot:GBQCQ}), and it is taken that the quantum output is encoded in these qubits. Note that in Protocol \ref{prot:GBQCC} we take the input to be $\ket + \otimes \cdots \otimes \ket + $, an encoding of the fixed classical input $0 \cdots 0$, any other arbitrary classical input $i_1 \cdots i_n$ is prepared by applying appropriate $Z$ on the corresponding qubit to create 
\begin{align*}
\ket e = Z_1^{i_1} \otimes \ldots \otimes  Z_n^{i_n} ( \ket {+_{\theta_1}} \otimes \cdots \otimes \ket {+_{\theta_n}}) \, .
\end{align*}
For classical input there is no need for a full one-time padding of the input hence no need for the $x_i$ random variables as $\theta_i$ rotation completely hides the input. The above explanation is the basis for the correctness of all of the protocols presented in this paper.

\begin{definition} A hiding protocol with quantum input is \emph{correct} if the quantum output state is $U\ket I$ or if the classical outputs are the result of Pauli $X$ measurements  on the state $U\ket I$, where $U$ is the unitary operator corresponding to the implementation of the measurement pattern of the hiding protocol. Similarly one could define correctness for protocols with classical input.
\end{definition}

\begin{theorem}[Correctness]
\label{thm:Correctness}Assume Alice and Bob follow the steps of Protocols~\ref{prot:GBQCQ} and \ref{prot:GBQCC}. Then the outcome is correct.
\end{theorem}
\begin{proof} The correctness of these protocol follows from the correctness of standard measurement based quantum computation~\cite{Mcal07}, as we now show. We explicitly give a proof only for the case of quantum input and output, as the remaining cases have virtually identical proofs. The protocol deviates in three ways from the standard implementation of the desired measurement pattern defined by a graph state $G$ with measurement angles $\phi_i$: a random $Z(\theta_i)$ rotation over all qubits; a random $X^{x_i}$ rotation over the input qubits; measuring with angles $\delta_i$. However, since  \textsc{ctrl}-$Z$ commutes with $Z$-rotations, Alice's preparation does not change the underlying graph state; only the phase of each qubit is locally changed, and it is as
if Bob had done the $Z$-rotation after the  \textsc{ctrl}-$Z$. Let $\phi_i'$ be the adapted angles of the measurement $\phi_i$ according to the flow structure of the desired measurement pattern defined by $G$. Note that a measurement in the $\{ |+_{\phi_i'} \rangle, |-_{\phi_i'} \rangle \}$
basis on a state $\ket{\psi}$ is the same as a measurement in the $\{ |+_{\phi_i'  +\theta_i}\rangle, |-_{\phi'_i  +\theta_i} \rangle \}$ basis on $Z(\theta_i) \ket{\psi}$. Also a measurement in the $\{| +_{\phi_i'}\rangle, |-_{\phi_i'}\rangle\}$
basis on a state $\ket{\psi}$ is the same as a measurement in the $\{|+_{- \phi_i'}\rangle, |-_{-\phi'_i}\rangle\}$ basis on $X \ket{\psi}$. Finally 
since $\delta_i = (-1)^{x_i}\phi'_i + \theta_i + \pi r_i$, if $r_i=0$, Bob's measurement
has the same effect as Alice's target measurement; if $r_i=1$,
all Alice needs to do is to flip the outcome. Therefore all the deviation from the actual implementation of the measurement pattern are corrected and the quantum output is the desired state corresponding to the action of the unitary operator implemented by the graph state $G$ over the input state.
\end{proof}

Note that in practice if Alice has the description of a unitary $V$ such that $V(\otimes_i \ket +) = \ket I$ then trivially a hiding protocol that blindly computes $UV$ over the input states $\otimes_i \ket +$ will prepare the desired output state of the form $U\ket I$. Therefore for such a scenario Alice can follow the step of the Protocol \ref{prot:GBQCQ} with classical input without having to prepare the encoded state $X_1^{x_1} Z_1(\theta_1) \otimes \ldots \otimes  X_n^{x_n} Z_n(\theta_n) \ket I$ herself. However, we have presented the full protocol for an arbitrary, possibly unknown, quantum input state, since the general scheme proved useful for dealing with input supplied by a third party \cite{Q-coins}.

\section{Blindness}

We say a hiding protocol is \emph{blind} if Bob cannot tell anything relating to the angles of measurements. In considering this it is worth noting that Bob can run the protocol only once with fixed values for Alice's parameters $\phi_i,\theta_i,r_i,x_i$. Later we will show how for generic graphs this will lead to hiding the output of the computation as well. Following the convention of \cite{AFK}, we use the notation of a leakage function, denoted as $L(X)$, to formalize what Bob learns during the interaction. We present a stand alone security definition that is equivalent to the original definition of blindness provided in \cite{BFK09}.

\begin{definition} \label{defn:blind} A hiding protocol \textsf{P} with input $X$ is \emph{blind} while leaking at most $L(X)$ if the distribution of messages obtained by Bob in \textsf{P} is dependent only on $L(X)$.
\end{definition}

\begin{theorem}[Blindness]\label{t-blind}
Protocol~\ref{prot:GBQCQ} is blind while leaking at most $G$ and $n$, and \ref{prot:GBQCC} is blind while leaking at most $G$.
\end{theorem}

\begin{proof}
We first give a proof for the blindness of Protocol~\ref{prot:GBQCQ}. We show that given $G$ and $n$, and independent of the actions of Bob, the message registers he receives are always in a maximally mixed state. We begin by introducing a new variable $\theta'_i = \theta_i + r_i \pi$, for all $i$. Thus, any quantum input received by Bob during a run of the protocol is given by $\ket e = X_1^{x_1} Z_1^{r_1} Z_1(\theta'_1) \otimes \ldots \otimes  X_n^{x_n} Z_n^{r_n} Z_n(\theta'_n) \ket I$, while the remaining qubits he receives are in states $|+_{\theta'_i+r_i\pi}\rangle$ for $n<i\leq m$. Expressed in terms of $\theta'_i$, $\delta_i$ becomes independent of $r_i$ for all $i$, since
 \begin{equation*}
\delta_i = (-1)^{s_{f^{-1}(i)}} \phi_i + \sum_{j: \, i\in N_G(f(j))} s_j \pi + \theta'_i.
\end{equation*}
Thus, only the $i$th qubit received by Bob is dependent on $r_i$, and so tracing over the secret values $r$ simply dephases every qubit in the computational basis. Similarly, only qubit $i$ is dependent on $x_i$ for $1\leq i \leq n$, and so tracing over $x$ completes the depolarization of the quantum input. Thus every qubit received by Bob is in the maximally mixed state, and uncorrelated with all other qubits.

Next consider the classical communication used to convey measurement angles during the protocol. The computation of $\delta_i$ is composed of three terms. The first two terms, $(-1)^{s_{f^{-1}(i)}} \phi_i$ and $\sum_{j: \, i\in N_G(f(j))} s_j \pi$, may depend implicitly on $b_k$ and $\delta_k$ for $k<i$, and on $r$ and $x$. However, note that the communication received up to Step $i$ is independent of $ \theta'_i$, the third term of $\delta_i$. Since $ \theta'_i$ is uniformly random over $A$, $\delta_i$ must also be uniformly random and uncorrelated with previous communication sent to Bob. Thus, all communication in the protocol is uniformly random and uncorrelated, once the random keys ($x,r,\theta$) are traced out, independent of the actions of Bob.  An identical argument holds for Protocol ~\ref{prot:GBQCC}, except that all $m$ qubits are assigned measurements, and hence $n$ is not revealed.
\end{proof}

We note that the above definition is equivalent to a simulator based definition, since once $L(X)$ is fixed, the distribution of messages Bob receives is also fixed. Hence, Alice could be replaced by a simulator with access only to $L(X)$, and this substitution could not be detected by Bob. A more detailed treatment of simulator based definitions and composable security can be found in \cite{DFPR13}.

\begin{algorithm}[h!]
\caption{Generic Hiding Protocol with Quantum Input and Output}
 \label{prot:GBQCQ}

\begin{itemize}
\item \textbf{Alice's resources} \\
\noindent -- Graph $G$ over $m$ vertices where labeling of vertices are in such a way that the first $n$ qubits are input and the last $n$ qubits are output. \\  
\noindent -- An $n$-qubit input state $\ket I$.\\
\noindent -- A sequence of non-output measurement angles, $\phi=(\phi_i)_{1\leq i \leq (m-n)}$ with $\phi_i \in A$. \\
\noindent -- $m$ random variables $\theta_i$ with values taken uniformly at random from $A$.\\
\noindent -- $n$ random variables $x_i$ and $m-n$ random variables $r_i$ with values taken uniformly at random from $\{0,1\}$. \\
\noindent -- A fixed function $C_G$ that for each non-output qubit $i$ ($1 \leq i \leq m-n$) computes the angle of the measurement of qubit $i$ to be sent to Bob. This function depends on $\phi_i,\theta_i,r_i,x_i$ and the result of the measurements that have been performed so far ($\mathbf{s}_{< i}$). The function $C_G$ also depends on the flow ($f,\preceq$) of the graph $G$. However, since the flow of the graph $G$ is unique (if it exists), we need not take flow as a parameter of the function $C_G$. We have 
\AR{
C_G: \{ 1, \cdots, (m-n)\} \times A \times A \times \{0,1\} \times \{0,1\}\times \{0,1\}^{m-n}  \rightarrow A \\ \\
(i, \phi_i, \theta_i, r_i, x_i, \mathbf{s}) \mapsto (-1)^{x_i+s_{f^{-1}(i)}} \phi_i + \sum_{j: \, i\in N_G(f(j))} s_j \pi + \theta_i + r_i \pi  \\
}
where $x_k$ for $n+1 \leq k \leq m$ and also $s_k$ for any non-defined value of $k$ is set to zero.

\item \textbf{Initial Step}

-- \textbf{Alice's move:} Alice sends Bob the graph $G$ and sets all the values in $\mathbf{s}$ to be $0$. Next she sends $m$ qubits in the order of the labeling of the vertices of the graph, as follows: first, Alice encodes the $n$-qubit input state as
\AR{
\ket e = X_1^{x_1} Z_1(\theta_1) \otimes \ldots \otimes  X_n^{x_n} Z_n(\theta_n) \ket I
}
and sends them as the first $n$ qubits to Bob. She then prepares $m-n$ single qubits in the state $\ket {+_{\theta_i}}$ ($n+1 \leq i \leq  m$) and sends them to Bob as the remaining qubits. \\
-- \textbf{Bob's move:} Bob receives $m$ single qubits and entangles them according to $G$.

\item \textbf{Step $i: \; 1 \leq i \leq (m-n)$}

-- \textbf{Alice's move:} Alice computes the angle $\delta_i=C_G(i, \phi_i, \theta_i, r_i, x_i, \mathbf{s})$ and sends it to Bob.\\ 
-- \textbf{Bob's move:} Bob measures qubit $i$ with angle $\delta_i$ and sends Alice the result $b_i$. \\
-- \textbf{Alice's move:} Alice sets the value of $s_i$ in $\mathbf{s}$ to be $b_i \oplus r_i$. \\

\item \textbf{Step $i: \; m-n+1 \leq i \leq m$}

-- \textbf{Bob's move:} Bob sends qubit $i$ to Alice. \\
-- \textbf{Alice's move:} Alice applies $X^{s_{f^{-1}(i)}}Z^{\sum_{j: \, i\in N_G(f(j))} s_j} Z(\theta_i)$ over qubit $i$. \\
\end{itemize}
\end{algorithm}

\begin{algorithm}[h!]
\caption{Generic Hiding Protocol with Classical Input and Output}
 \label{prot:GBQCC}

\hskip 0.2cm 

\begin{itemize}
\item \textbf{Alice's resources} \\

\noindent -- Graph $G$ over $m$ vertices where labeling of vertices are in such a way that the first $n$ qubits are input and the last $n$ qubits are output. \\  
\noindent -- An $n$-bit input string $c_1,\ldots,c_n$.\\
\noindent -- A sequence of non-output measurement angles, $\phi=(\phi_i)_{1\leq i \leq (m-n)}$ with $\phi_i \in A$. \\
\noindent -- $m$ random variables $\theta_i$ with values taken uniformly at random from $A$.\\
\noindent -- $m$ random variables $r_i$ with values taken uniformly at random from $\{0,1\}$. \\
\noindent -- A fixed function $C_G$ that for each non output qubit $i$ ($1 \leq i \leq m$) computes the angle of the measurement of qubit $i$ to be sent to Bob:  
\AR{
C_G: \{ 1, \cdots, m\} \times A \times A \times \{0,1\} \times \{0,1\}^{m}  \rightarrow A \\ \\
(i, \phi_i, \theta_i, r_i, \mathbf{s}) \mapsto (-1)^{s_{f^{-1}(i)}} \phi_i + \sum_{j: \, i\in N_G(f(j))} s_j \pi + \theta_i + r_i \pi  \\
}
where $s_k$ for any non-defined value of $k$ is set to zero, also $\phi_i=0$ for $m-n+1 \leq i \leq m$.
\item \textbf{Initial Step}

-- \textbf{Alice's move:} Alice sends Bob the graph $G$ and sets all the value in $\mathbf{s}$ to be $0$. Next she sends $m$ qubits in the order of the labeling of the vertices of the graph, as follows: first, Alice encodes the $n$-bit string classical input $c_1,\ldots,c_n$ as state
\AR{
\ket e = Z_1^{c_1} \otimes \ldots \otimes  Z_n^{c_n} ( \ket {+_{\theta_1}} \otimes \cdots \otimes \ket {+_{\theta_n}})~= \ket {+_{\theta_1 + i_1 \pi}} \otimes \cdots \otimes \ket {+_{\theta_n+i_n \pi}}}
and sends them as the first $n$ qubits to Bob. She then prepares $m-n$ single qubits in the state $\ket {+_{\theta_i}}$ ($n+1 \leq i \leq  m$) and sends them to Bob as the remaining qubits. \\
-- \textbf{Bob's move:} Bob receives $m$ single qubits and entangles them according to $G$.
\item \textbf{Step $i: \; 1 \leq i \leq m$}
-- \textbf{Alice's move:} Alice computes the angle $\delta_i=C_G(i, \phi_i, \theta_i, r_i, \mathbf{s})$ and sends it to Bob.\\ 
-- \textbf{Bob's move:} Bob measures qubit $i$ with angle $\delta_i$ and sends Alice the result $b_i$. \\
-- \textbf{Alice's move:} Alice sets the value of $s_i$ in $s$ to be $b_i\oplus r_i$. \\

\end{itemize}
\end{algorithm}

\section{Dummy Qubits}\label{s-dummy}

In order to obtain an intuitive method for achieving verification, we construct an extension of Protocol \ref{prot:GBQCQ} where Alice can also prepare qubits in the state $\ket z$ where $z$ is chosen uniformly at random from $\{0,1\}$. These qubits are called \emph{dummy qubits}, as they will not be part of actual computation. A dummy qubit remains disentangled from the rest of the qubits of the graph state and, as we prove later, the addition of these dummy qubits does not affect the correctness or blindness of the hiding protocol. These dummy qubits are measured with random angles which again will not affect the actual computation due to the fact that they are disentangled from the rest of the qubits. However, as we demonstrate in the next section, these dummy qubits allow Alice to easily create isolated trap qubits within the resource state to enable verification of the computation. Note that Alice must keep the position of the dummy qubits hidden from Bob (i.e. part of the secret) in order to keep the position of any trap qubits hidden. The addition of the dummy qubits can also be viewed as a method for the blind implementation of the Pauli $Z$ basis measurements. This is due to the fact that their position is hidden from Bob and from his point of view they are measured in the $(X,Y)$ plane as well. However due to their preparation state ($\ket{0}$ or $\ket 1$) through the entangling step, they have the same effect of measuring the corresponding qubit in the Pauli $Z$ basis. Therefore, we use the term blind Pauli $Z$ measurement interchangeably with dummy qubits in the rest of the paper. Due to the addition of dummy qubits, we will assume from now on that $n$ is an upper bound over the number of the input or output qubits. This is required to allow the possibility of having hidden trap or dummy qubits as part of the input or output system. Therefore in the design of the measurement pattern, auxiliary qubits are added to the input and output space in such a way that the actual computation remains intact.

\begin{algorithm}
\caption{Generic Hiding Protocol with Quantum Input and Output and Dummy Qubits}
 \label{prot:DBQC}

\begin{itemize}
\item \textbf{Alice's resources} \\
{\noindent -- Graph $G$ over $m$ vertices where labeling of vertices are in such a way that all the $l$ input qubits are located among the first $n \geq l$ qubits and all the $l$ output qubits are located among the last $n$ qubits. \\
\noindent -- An $l$-qubit input state $\ket I$. \\
\noindent -- The dummy qubits positions, set $D$, chosen among all possible vertices except the $l$ input and $l$ output qubits.\\}
\noindent -- A sequence of non-output measurement angles, $\phi=(\phi_i)_{1\leq i \leq (m-n)}$ with $\phi_i \in A$ where $\phi_i = 0$ for all $i\in D$.\\
\noindent -- $m$ random variables $\theta_i$ with values taken uniformly at random from $A$.\\
\noindent -- $l$ random variables $x_i$, $m-n$ random variables $r_i$ and $|D|$ random variables $d_i$ with values taken uniformly at random from $\{0,1\}$. \\
\noindent -- A fixed function $C_G$ that for each non output qubit $i$ ($1 \leq i \leq m-n$) computes the angle of the measurement of qubit $i$ to be sent to Bob:
\AR{
C_G: \{ 1, \cdots, (m-n)\} \times A \times A \times \{0,1\} \times \{0,1\} \times \{0,1\}^{m-n}  \rightarrow A \\ \\
(i, \phi_i, \theta_i, r_i, x_i, \mathbf{s}) \mapsto (-1)^{x_i+s_{f^{-1}(i)}} \phi_i + \sum_{j: \, i\in N_G(f(j))} s_j \pi + \theta_i + r_i \pi  \\
}
where $x_k$ for $n+1 \leq k \leq m$ and $s_k$ for any non-defined value of $k$ are set to zero. 
\item \textbf{Initial Step}

-- \textbf{Alice's move:} Alice sends Bob the graph $G$ and sets all the value in $\mathbf{s}$ to be $0$. Alice encodes the $l$-qubit input state as
\AR{
\ket e = X_1^{x_1} Z_1(\theta_1) \otimes \ldots \otimes  X_n^{x_l} Z_n(\theta_l) \ket I
}
and positions them among the first $n$ qubits. She then prepares the remaining qubits in the following form
\AR{
\forall i\in D &\;\;\;& \ket {d_i} \\ 
\forall i \not \in D &\;\;\;& \prod_{j\in N_G(i) \cap D} Z^{d_j}\ket {+_{\theta_i}}~=\ket {+_{\theta_i + \sum_{j\in N_G(i) \cap D} d_j \pi}}}

Then Alice sends Bob all $m$ qubits in the order of the labeling of the vertices of the graph.

-- \textbf{Bob's move:} Bob receives $m$ single qubits and entangles them according to $G$.

\item \textbf{Step $i: \; 1 \leq i \leq (m-n)$}

-- \textbf{Alice's move:} Alice computes the angle $\delta_i=C_G(i, \phi_i, \theta_i, r_i, \mathbf{s})$ and sends it to Bob.\\ 
-- \textbf{Bob's move:} Bob measures qubit $i$ with angle $\delta_i$ and sends Alice the result $b_i$. \\
-- \textbf{Alice's move:} Alice sets the value of $s_i$ in $\mathbf{s}$ to be $b_i \oplus r_i$. \\

\item \textbf{Step $i: \; m-n+1 \leq i \leq m$}

-- \textbf{Bob's move:} Bob sends qubit $i$ to Alice. \\
-- \textbf{Alice's move:} Alice applies $X^{s_{f^{-1}(i)}}Z^{\sum_{j: \, i\in N_G(f(j))} s_j} Z(\theta_i)$ to qubit $i$. \\
\end{itemize}
\end{algorithm}

\begin{theorem}
\label{thm:CorrectnessDummy}Assume Alice and Bob follow the steps of
Protocol~\ref{prot:DBQC}. Then the outcome obtained is the same as if the computation took place over the graph $G$ after removal of the dummy vertices in $D$, the set of positions of dummy qubits in $G$. 
\end{theorem}
\begin{proof}
The proof is similar to the proof of Theorem \ref{thm:Correctness}, the only new element is the effect of the dummy qubits. If a dummy qubit is in the state $\ket 0$ then in the entangling step this qubit does not affect the state of the other qubits. However, if the dummy qubit is in the state $\ket 1$ then the entangling operation will introduce a Pauli $Z$ rotation on all the neighboring qubits in $G$. Hence a qubit $i\not \in D$ will be affected by the operator $\prod_{j\in N_G(i) \cap D} Z^{d_j}$. In the initial step, Alice already applied the operation $\prod_{j\in N_G(i) \cap D} Z^{d_j}$ over the prepared qubits and therefore all qubits $i\not \in D$ are in the desired state $\ket {+_{\theta_i}}$, since $Z$ operator is self-inverse. Moreover all the dummy qubits are unentangled with the rest of qubits and are measured in a random basis with no consequences for the part of the computation taking place over the graph $G$ after removing vertices $D$. 
\end{proof}

\begin{theorem}\label{t-DBQC}
The hiding protocol with dummy qubits, Protocol \ref{prot:DBQC}, is blind while leaking at most $G$.
\end{theorem}

\begin{proof}
Proof follows along similar lines of Theorem \ref{t-blind}. We define $\theta'_{i} = \theta_{i} + \pi r_{i} + \pi \sum_{j\in N_G(i) \cap D} d_i$. Alice's total communication to Bob consists of the initial quantum states, which we can rewrite as $|+_{\theta'_{i}-\pi r_{i}}\rangle$ if the qubit is not a dummy qubit or $\in_R \{\ket{0},\ket{1}\}$ if it is a dummy qubit, and the measurement angles which are set to be $\delta_{i} = \phi'_{i} + \theta'_{i} - \pi \sum_{j\in N_G(i) \cap D} d_i$. As before, the values of $\delta_i$ are uniformly random since $\theta_i'$ are uniformly random, and for any fixed values of $\delta_i$ 
tracing over all $r_i$, we obtain the initial quantum state for each qubit as either 
\begin{align*}
\frac{1}{2} | +_{\theta'_{i}} \rangle \langle +_{\theta'_{i}} | + \frac{1}{2} | -_{\theta'_{i}} \rangle \langle -_{\theta'_{i}}| = \frac{\mathbb{I}}{2}
\end{align*}
if the qubit was not dummy, and 
\begin{align*}
\frac{1}{2}\ket{0}\bra{0} + \frac{1}{2}\ket{1}\bra{1} = \frac{\mathbb{I}}{2}
\end{align*}
if the qubit was a dummy. Hence the qubits obtained by Bob are always in the maximally mixed state and are not correlated with each other. 
\end{proof}

\section{Universal Resource States} \label{s-resource}

During a hiding protocol Bob learns the graph of entanglement, $G$, however it was shown in \cite{BFK09} that it is possible for Alice to choose a family of graphs corresponding to what were termed \emph{brickwork states} such that blindness of the angles, as defined before, will permit Alice to hide the unitary operator that the protocol is implementing, revealing only an upper bound on the dimensions of the circuit required to implement it. The key element to achieve this is the use of those universal resources for MBQC \cite{NDMB07} that are generic, hence revealing no information about the structure of the underlying computation, except the bounds on the size of input and the depth of the computation. Moreover to make the protocol practical from Alice's point, it is desirable to restrict the class of measurement angles, so that the required class of random qubits prepared by Alice is also restricted. Note that exact universal blind quantum computing could be achieved if Alice could prepare separable single qubit states $\ket {+_{\theta}}$ with $\theta$ chosen randomly in $[0,2\pi)$ and if Bob could make any measurement with angles in $[0,2\pi)$. Such a model requires Alice to communicate random real angles to Bob, and hence such a setting is unattractive from a communications resources point of view. Similar to the quantum circuit scenario, by the Solovay-Kitaev theorem, a finite set of angles (for instance a set that corresponds to Hadamard and $\frac{\pi}{8}$-Phase gates) can be used to efficiently approximate any single qubit unitary operator.\footnote{More precisely, the Solovay-Kitaev theorem states that if the subgroup generated by some subset of $SU(2)$ operators is dense in $SU(2)$, then the approximation converges exponentially quickly to any element of $SU(2)$ in the number of these operators from a smaller set one uses to approximate.} For the rest of this paper we will restrict our attention to approximate universality and we use the fact that a large family of graph states are approximately universal if one restricts the set of angles to be in the set $\{0, \pm \pi/4,  \pm \pi/2\}$ \cite{generator04}. We give two such examples below.

\begin{definition} \label{defn:brick} A brickwork state $\mathcal{G}_{n \times m}$, where $m \equiv 5 \pmod 8$, is an entangled state of~$n\times m$ qubits constructed as follows:
\begin{enumerate}
\item Prepare all qubits in state $\ket +$ and assign to each qubit an index $(i,j)$, $i$ being a row ($i \in [n]$) and $j$ being a column ($j \in [m]$).
\item For each row, apply the operator \textsc{ctrl}-$Z$ on qubits $(i,j)$ and $(i,j+1)$ where $1\leq j \leq m-1$.
\item For each column $j \equiv 3 \pmod 8$ and each odd row $i$, apply the operator \textsc{ctrl}-$Z$ on qubits $(i,j)$ and $(i+1,j)$ and also on qubits $(i,j+2)$ and $(i+1,j+2)$.
\item For each column $j \equiv 7 \pmod 8$ and each even row $i$, apply the operator \textsc{ctrl}-$Z$ on qubits $(i,j)$ and $(i+1,j)$ and also on qubits $(i,j+2)$ and $(i+1,j+2)$.
\end{enumerate}
We will refer to the underlying graph of a brickwork state as the brickwork graph and denote it with the same notation as $\mathcal{G}_{n \times m}$, see Figure \ref{fig:brick}. 
\end{definition}
\begin{figure}[h]
\begin{center}
\includegraphics[width=0.75\columnwidth]{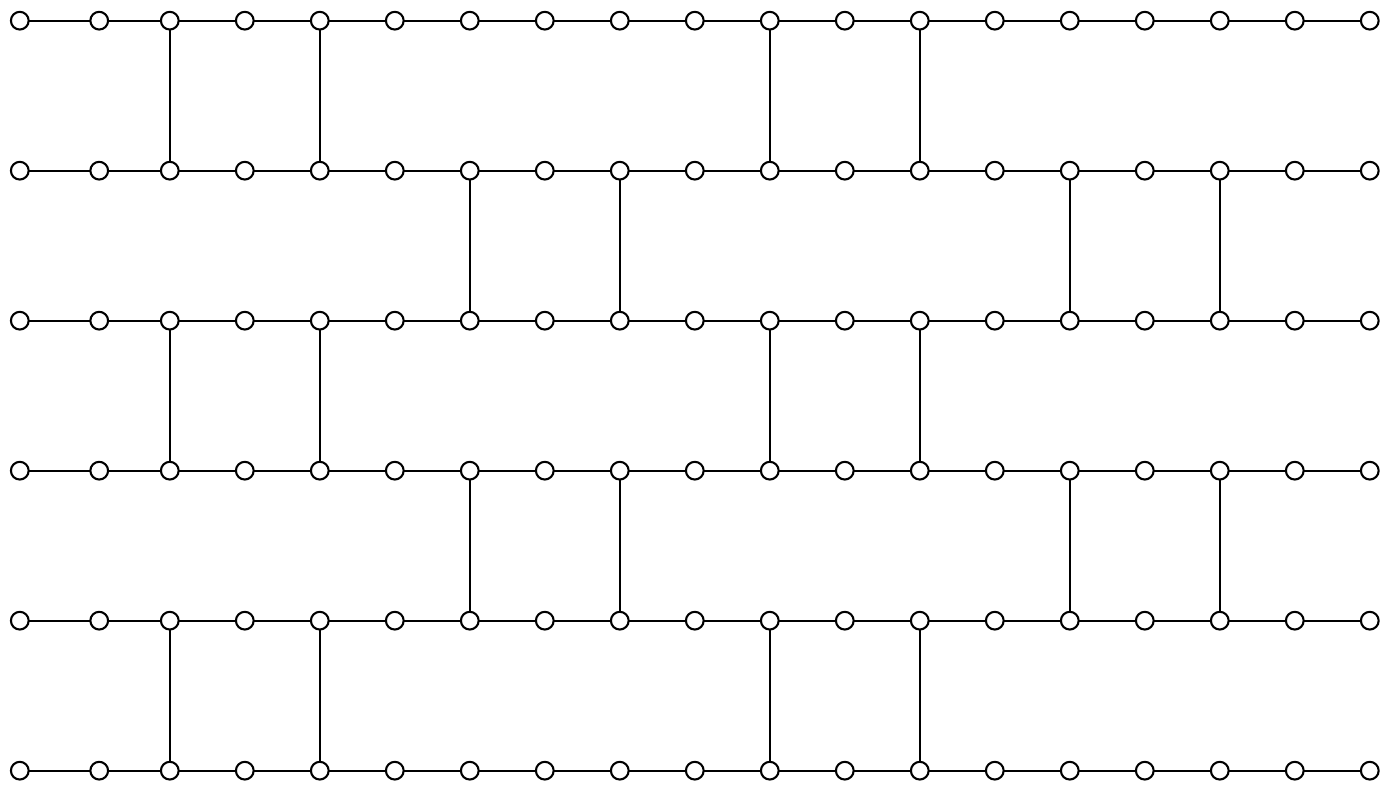}
\caption{The brickwork state, $\mathcal{G}_{6 \times 19}$. Qubits are arranged according to layer $x$ and row $y$, corresponding
to the vertices in the above graph, and are originally in state $\ket{+}$. \textsc{ctrl}-$Z$ gates are then performed between qubits which are joined by an edge. A similar resource state was proposed in \cite{CLN05}.}
\label{fig:brick}
\end{center}
\end{figure}

\begin{theorem}[Universality \cite{BFK09}]
\label{thm:universal} The \graphstate $\mathcal{G}_{n \times m}$ is
universal for quantum computation. Furthermore, we only require
single-qubit measurements under the angles $\{0, \pm \pi/4,  \pm \pi/2\}$ to achieve approximate universality,
and measurements can be done layer-by-layer.
\end{theorem}

\begin{proof}
The proof is straightforward (see details in \cite{BFK09}) based on constructing measurement patterns for elements of an approximate universal gates set that could be tiled together as a brickwork states as depicted in Figures \ref{Fig:GeneralRotation}.
\end{proof}

\begin{figure} \centering
   \includegraphics[width=\columnwidth]{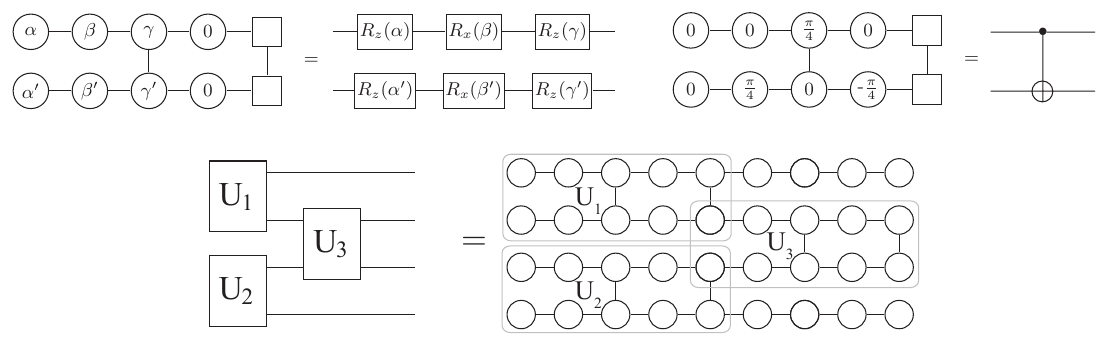}
\caption{Measurement patterns implementing arbitrary single qubit rotations and the CNOT operator. These patterns can be composed within the brickwork state, as shown in the lower portion of the figure.}
    \label{Fig:GeneralRotation}
\end{figure}

Let us denote vertices of a brickwork graph $\mathcal{G}_{n\times m}$ by $(i,j)$ (where $1\leq i \leq n, 1 \leq j \leq m$), then it is easy to verify that the unique flow function of $\mathcal{G}$ is defined by: 
\begin{align*}
f_{\mathcal{G}}( (i,j) ) =  (i,j+1) 
\end{align*}
That is to say, the flow of each vertex in the graph is from its immediate left neighbor in the same row. The corresponding partial order $\prec_{\mathcal{G}}$ is defined as the collection of sets $L_j$ of all vertices in the $j$th column of the brickwork graph
\begin{align*}
L_j = \{(x,y)| 1\leq x \leq n, y = j\}.
\end{align*}
Now suppose Alice has in mind a unitary operator $U$ of size $2^n \times 2^n$ and the $n$-qubit input state $\ket I$. Due to Theorem \ref{thm:universal} there exist an integer $m$ and angles $\{ \phi_{i,j} \}_{1\leq i \leq n, 1 \leq j \leq m} \in A$ such that the measurement pattern with angles $\{\phi_{i,j}\}$ over the brickwork state $\mathcal{G}_{n \times m}$, where the first $n$ qubit are set to be in the state $\ket I$, approximates ~$U\ket I$. Therefore the last $n$ qubits after the measurements of the first $m-n$ qubits and application of the corresponding corrections induced by flow are in a state which can be made arbitrarily close to $U\ket I$. We can simply adapt the generic hiding protocol to implement this measurement pattern blindly as presented in the \cite{BFK09}. 

As mentioned in Section \ref{s-dummy}, in order to construct a verification scheme we make use of dummy qubits. While this presents a simple mechanism to achieve isolated trap qubits, the presence of trap and dummy qubits disrupts the computation. However this can be fixed through a simple modification of the brickwork state.

\begin{definition} \label{defn:brickcylinder} A cylinder brickwork state $\mathcal{G}^C_{n \times m}$ is a modification of the brickwork state of size ${n \times m}$, for even $n$, where the first and the last rows are connected such that the regular brickwork structure is preserved, while introducing rotational symmetry. We will refer to the underlying graph of a cylinder brickwork state as the cylinder brickwork graph and denote it with the same notation as $\mathcal{G}^C_{n \times m}$ (see Figure \ref{fig:cylinder}). A \emph{tape}, $\mathcal{T}_i$ in a cylinder brickwork graph is the subgraph induced by all the nodes of $i$th and $(i+1)$th rows.
\end{definition}

\begin{figure}[h]
\begin{center}
\scalebox{0.5}[0.5]{\includegraphics{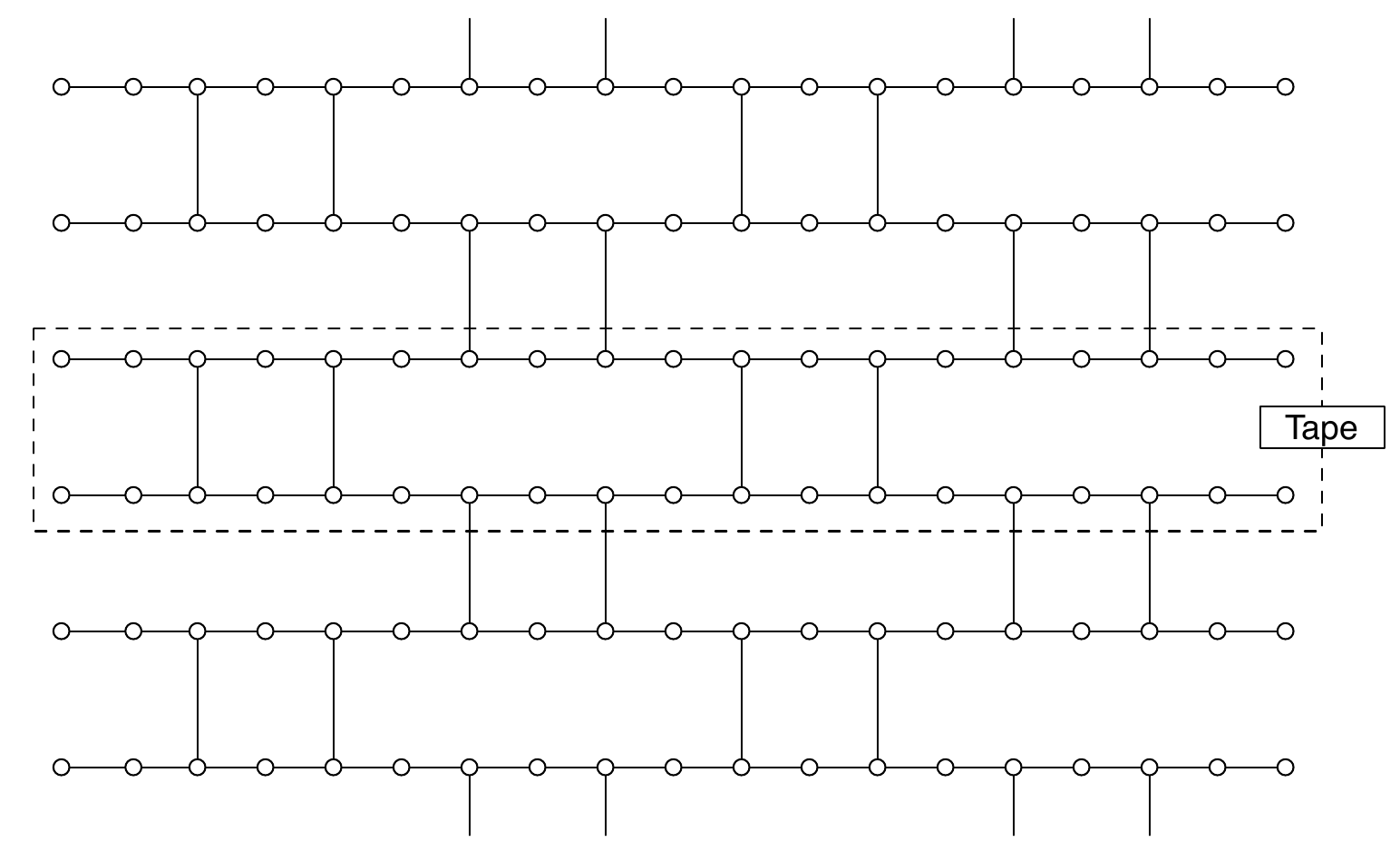}}
\caption{Th cylinder brickwork state $\mathcal{G}^C_{6 \times 19}$.}
\label{fig:cylinder}
\end{center}
\end{figure}

The cylinder brickwork state allows for a simple construction for trap-based verification, as discussed in Section \ref{s-verification}. Next we introduce another generic family called \emph{dotted-complete graph states} which enables significant amplification of the probability of detecting deviations from the computation, particularly in the case of quantum output, as discussed in Section \ref{s-amplification}. The basic idea behind this new universal resource state is that it can be partitioned blindly into smaller universal resource states, one of which will be used for the computation, while the others will be used as traps for verification purposes (see later). To begin with, we need to introduce the graphs which we will use, and prove that they have some special properties.

\begin{definition} 
We define the operator $\sim(G)$ on graph $G$ to be the operator which transforms a graph $G$ to a new graph denoted as $\tilde{G}$ by replacing every edge in $G$ with a new vertex connected to the two vertices originally joined by that edge. Let $K_N$ denote the complete graph of $N$ vertices, we call the quantum state corresponding to the graph $\tilde{K}_N$ the \emph{dotted-complete graph state} denoted with $\tilde{\mathcal{K}}_N$. We denote the set of vertices of $\tilde{K}_N$ previously inherited from $K_N$ as $P(\tilde{K}_N)$, and denote the vertices added by the $\sim()$ operation by $A(\tilde{K}_N)$. The number of the vertices in the $\tilde{K}_{N}$ graph is then equal to $N (N+1)/2$.
\end{definition}

\begin{figure}[h!]
\begin{center}
\includegraphics{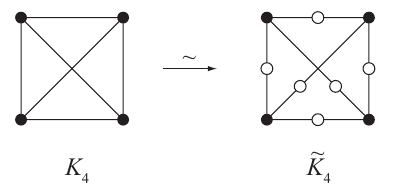}
\end{center}
\caption{An example of the relationship between a complete graph $K_4$ and the corresponding dotted-complete graph $\tilde{K}_4$. The vertices in black in $\tilde{K}_4$ denote the set $P(\tilde{K}_4)$, while the white vertices correspond to $A(\tilde{K}_4)$.}
\end{figure}
The next definition and lemmas will be used in manipulation of dotted-complete graph states.  

\begin{definition} 
We define the \emph{bridge} operator on a vertex $v$ of degree 2 on graph $G$ to be the operator which connects the two neighbors of $v$ and then removes vertex $v$ and any associated edges from $G$. We define the \emph{break} operator on a vertex $v$ of graph $G$ to be the operator which removes vertex $v$ and any associated edges from $G$. Let $G$ be a graph on $m$ vertices. Then we say that $G$ is \emph{$n$-universal}, for $n\leq m$, if and only if any graph of $n$ vertices can be obtained from $G$ through a sequence of bridges and breaks. 
\end{definition}

\begin{lemma}\label{lem:m-universal}
$\tilde{K}_N$ is \emph{$N$-universal}, and the bridge and break operations used to obtain a target graph need only be performed on vertices in $A(\tilde{K}_N)$.
\end{lemma}
\begin{proof}
Given any graph $G$  on $N$ vertices, associate each vertex $u_i$ in $G$ with a vertex $v_i$ in $P(\tilde{K}_N)$. Each pair of vertices $(v_i,v_j)$ in $P(\tilde{K}_N)$ is connected through an intermediate vertex of degree 2 in $A(\tilde{K}_N)$. Thus by bridging over the intermediate vertex if $u_i$ and $u_j$ are joined by an edge and breaking the intermediate vertex otherwise, $\tilde{K}_N$ reduces to $G$. As this is true for all graphs $G$ on $N$ vertices, $\tilde{K}_N$ is $N$-universal.
\end{proof}

\begin{lemma}\label{lem:partition}
Given a partitioning of the vertices $P(\tilde{K}_N)$ into $n$ sets $\{P_i\}$ containing $N_i$ vertices respectively, by applying a sequence of break operations only, it is possible to transform $\tilde{K}_N$ into $n$ disconnected graphs $\tilde{k}_i$ such that each one of them are of the form  $\tilde{K}_{N_i}$ and $P(\tilde{k}_i) = P_i$. 
\end{lemma}

\begin{proof}
As the vertices $P(\tilde{K}_N)$ are associated with a corresponding vertex in $K_N$, the vertices of $K_N$ can by partitioned into the sets $\{P_i\}$. As $K_N$ is the complete graph the vertices within each partition $P_i$ form a clique. Thus by removing edges between the partitions the resulting graph is composed of $n$ disconnected graphs $\{k_i = K_{N_i}\}$ such that the vertices in $k_i$ are the vertices in $P_i$. As removing an edge before applying the $\sim()$ operator is equivalent to applying a break operation after the $\sim()$ operator there exists a corresponding sequence of break operations, such that the resulting graph is $\sim(\{k_i\})= \{\tilde{k}_i\}$. As $\tilde{k}_i=\,\sim(k_i)$, it follows that $P(\tilde{k}_i)=P_i$ and since $k_i = K_{N_i}$ then $\tilde{k}_i = \tilde{K}_{N_i}$ as required. 
\end{proof}

\begin{lemma}\label{lem:p-removal}
Given a graph $\tilde{K}_N$, by applying break operators to every vertex in $P(\tilde{K}_N)$ or $A(\tilde{K}_N)$ the resulting graph is composed of the vertices of $A(\tilde{K}_N)$ or $P(\tilde{K}_N)$ respectively and contains no edges.
\end{lemma}
\begin{proof}
As the $\sim()$ operation only introduces vertices connected to vertices in $P(\tilde{K}_N)$, every vertex in $A(\tilde{K}_N)$ shares edges only with vertices in $P(\tilde{K}_N)$. Thus when the vertices in $P(\tilde{K}_N)$ and their associated edges are removed by the break operators, the vertices in $A(\tilde{K}_N)$ become disconnected. Similarly, since $\sim()$ removes all edges between vertices in $P(\tilde{K}_N)$, hence every vertex in $P(\tilde{K}_N)$ shares edges only with vertices in $A(\tilde{K}_N)$. Thus when the vertices in $A(\tilde{K}_N)$ and their associated edges are removed by the break operators, the vertices in $P(\tilde{K}_N)$ become disconnected.
\end{proof}

We now extend these results to graph states.

\begin{lemma}\label{lem:state-graph}
Given two graph states $\ket{\psi_{G_1}}$ and $\ket{\psi_{G_2}}$ corresponding to graphs $G_1$ and $G_2$ respectively, if it is possible to obtain $G_2$ from $G_1$ through a sequence of bridge and break operations, then it is possible to obtain $\ket{\psi_{G_2}}$ from $\ket{\psi_{G_1}}$ through a sequence of Pauli measurements and local rotations about the $Z$ axis through angles from the set $\{0,\frac{\pi}{2},\pi,\frac{3\pi}{2}\}$.
\end{lemma}

\begin{proof}
By measuring any qubit in a graph state with Pauli $Z$ operator, we obtain a state equivalent up to local Pauli $Z$ corrections to the graph state obtained from the graph when that vertex and its associated edges are removed. To see this, we consider the operations this qubit undergoes: It is first prepared in a state $\ket{+}$, then interacted with its neighbors via \textsc{ctrl}-$Z$ gates, and then measured in the $Z$ basis. As the measurement commutes with the entangling operation, this result is identical to the case where the \textsc{ctrl}-$Z$ gates are applied to the measured eigenstate of $Z$. Thus when the complete sequence of events is taken into account, this operation is equivalent to the identity when the measurement outcome is $0$, and equivalent to local Pauli $Z$ operators applied to the neighbors of the measured site when the measurement outcome is $1$. This is then the graph state equivalent of the break operation defined on the associated graph. 

If a vertex is of degree 2, then measuring the associated qubit with the Pauli $Y$ operator yields the graph state corresponding to the graph obtained by applying a bridge operation to that vertex, up to local $Z$-rotations through an angle $\pm \frac{\pi}{2}$. To see this, we again consider the sequence of operations the qubit undergoes: It is prepared in the state $\ket{+}$, interacted with its neighbors and then measured in the $Y$ basis. Immediately prior to measurement, the net operator applied is $\frac{1}{\sqrt{2}}\ket{0}\otimes\mathbb{I} + \frac{1}{\sqrt{2}}\ket{1}\otimes{Z_1 \otimes Z_2}$, where the subscripts 1 and 2 denote the neighbors of the measured qubit. Thus if the measurement result is $0$ then this is equivalent to directly applying the operator $e^{i\frac{\pi}{4}Z_1\otimes Z_2}$ to the neighboring qubits, whereas if the measurement result is 1 this is equivalent to applying the operator $e^{-i\frac{\pi}{4}Z_1\otimes Z_2}$ to these qubits. Since the \textsc{ctrl}-$Z$ gate can be written either as $e^{i\frac{\pi}{4}(\mathbb{I} - Z\otimes \mathbb{I} - \mathbb{I}\otimes Z + Z\otimes Z)}$ or $e^{-i\frac{\pi}{4}(\mathbb{I} - Z\otimes \mathbb{I} - \mathbb{I}\otimes Z + Z\otimes Z)}$, the effect on the neighboring qubits is equivalent to a \textsc{ctrl}-$Z$, up to local $Z$-rotations by $\frac{\pi}{2}$ (for a measurement result of 0) or $-\frac{\pi}{2}$ (for a measurement result of 1). This could also be derived via the stabilizer formalism. For a more detailed discussion of the effect of Pauli measurements in the measurement based model, the reader is referred to \cite{HEB04}.
\end{proof}

\begin{theorem}[Universality]
\label{thm:universal2} The dotted-complete graph state $\tilde{\mathcal{K}}_N$ is
universal for quantum computation. Furthermore, we only require
single-qubit measurements under the angles $\{0, \pm \pi/4,  \pm \pi/2\}$ and in the Pauli $Z$ basis to achieve approximate universality,
and measurements can be done layer-by-layer.
\end{theorem}
\begin{proof}
Due to lemmas \ref{lem:m-universal} and \ref{lem:state-graph}, by choosing $N$ big enough, we could construct the brickwork state $\mathcal{G}_{n \times m}$ from $\tilde{\mathcal{K}}_N$ using only Pauli measurements. Hence from Theorem \ref{thm:universal} we obtain the universality of dotted-complete graph states and approximate universality with only single qubits measurements under the angles $\{0, \pm \pi/4,  \pm \pi/2\}$ (which includes the Pauli $Y$ measurements required to implement bridge operations), and the Pauli $Z$ basis measurements required to implement break operations.
\end{proof}

From this result we can construct a new universal hiding protocol based on dotted-complete graph states, as given in Protocol \ref{prot:UBQC2}. Interestingly, in the case of classical input and output this new protocol does not even reveal the circuit dimensions, but instead a single integer which is an upper bound on the number of qubits required to implement the computation in the measurement-based model.

\begin{algorithm}
\caption{Dotted-Complete Graph State Universal Hiding Protocol with Quantum Input/Output}
 \label{prot:UBQC2}
\begin{itemize}
\item \textbf{Alice's resources} \\
\noindent -- Parameter $N$ such that the desired computation could be obtained from the state $\tilde{\mathcal{K}}_N$ after a sequence of break and bridge operators (Theorem \ref{thm:universal2}). The labeling of vertices are in such a way that the first $n$ qubits are input and the last $n$ qubits are output. \\  
\noindent -- The dummy qubits position, set $D$, is set to be the position of all the qubits that are required to be Pauli $Z$ measured for performing the break operators. \\
\noindent -- A sequence of non-output measurement angles, $\phi=(\phi_i)_{1\leq i \leq (m-n)}$ with $\phi_i \in A$ where $\phi_i = \frac \pi 2$ for all $i\in D$ and also for all the qubits that are required to be Pauli $Y$ measured to perform the bridge operators.\\
\noindent -- The rest of the resources are the same as Protocol \ref{prot:DBQC}.
\end{itemize}
\vskip 0.2cm
Follow the steps of Protocol \ref{prot:DBQC} where $G$ is replaced with $\tilde{K}_{N}$.
\vskip 0.2cm
\end{algorithm}

\begin{theorem}\label{t-blindinputdotted} Protocol \ref{prot:UBQC2} is blind, while leaking at most $n$ and $N$. 
\end{theorem}
\begin{proof}
As Bob entangles according to $\tilde{\mathcal{K}}_N$, clearly the parameter $N$ is leaked. Additionally, in the case of quantum output, Bob must be instructed how many qubits to return to Alice, and hence knows $n$. However, fixing these parameters, due to Theorem \ref{t-blind} all the measurement angles including the measurements for the bridge operators are blind to Bob. Similarly, from Theorem \ref{t-DBQC} we have blindness for the measurement corresponding to the break operators. Together these guarantee the blindness of the operations required to prepare a brickwork state from $\tilde{\mathcal{K}}_N$. Finally Theorem \ref{t-blind} proved the blindness of the remaining measurements performed on the prepared brickwork state.
\end{proof}

\section{Verification}\label{s-verification}

This section deals with another property of the hiding protocol called verification. This property requires that Alice can verify with high probability whether Bob has followed the instructions of the protocol and hence if the quantum or classical output state is indeed in the correct form, or whether there has been a deviation and she should therefore reject the output state. The main idea is to exploit blindness so that Alice can expand the protocol to include \emph{trap qubits} where Alice knows in advance the classical outcome of these specific measurements (i.e. the correct message from Bob for these measurements), where the blindness ensures that the position of these traps remains hidden from Bob. At the end Alice will accept the quantum or classical output only if Bob has produced all of the \emph{expected} outcomes for these trap qubits measurements. The subtlety in verification is to prove that the accepted quantum or classical output is indeed correct.

It is essential that Alice keeps the position of these trap qubits unknown to Bob, so that he cannot attempt to interfere with the actual computation of $U$ while keeping the trap qubits untouched. We will present a protocol where every qubit of the underlying graph could potentially be an isolated (unentangled) trap qubit in an unknown state $\ket {+_{\theta}}$ for $\theta \in A$. In order to do so, it is enough to prepare all the neighboring vertices of the trap qubit as dummy qubits, hence these dummy qubits together with the trap qubits remain disentangled from the rest of the graph during the preparation stage. Building on this simple construction, by adding more traps and adding error detection elements, we will present a final protocol in which the probability of not detecting an incorrect outcome is exponentially small. 

In order to first demonstrate the main idea of this method of verification, we ignore the universality property and only later will we present a concrete universal blind quantum computing protocol with the verification property. Hence to obtain a generic hiding protocol with a random unknown trap it is sufficient to use Protocol \ref{prot:DBQC}, where Alice chooses a random position $t$ to be an isolated trap qubit (Protocol \ref{prot:DTBQC}). 

\begin{algorithm}[h]
\caption{Generic Hiding QC For Unitary with Dummy, Trap, Quantum Input and Output}
 \label{prot:DTBQC}
 \vskip 0.2 cm
\begin{itemize}
\item \textbf{Alice's resources} \\
\noindent -- Graph $G$ over $m$ vertices and a random position $t$ among the vertices of $G$. \\  
\noindent -- The rest of the resources are the same as Protocol \ref{prot:DBQC} where $\phi_{i} = 0$ for $i=t$ and $i\in D$ where $D$ contains the set of all neighbors of position $t$ in the original graph to create an isolated trap qubit at position $t$. 

\item \textbf{Follow the steps of Protocol \ref{prot:DBQC}.}

\item \textbf{Accept/Reject} \\
\noindent --  After obtaining all the output qubits from Bob, if the trap qubit, $t$, is an output qubit, Alice measures it with angle $\delta_{t} = \theta_t + r_t \pi$ to obtain $b_t$. \\
\noindent -- Alice accepts if $b_{t} = r_{t}$. 

\end{itemize}
\end{algorithm}

Theorem \ref{t-DBQC} directly implies that Protocol \ref{prot:DTBQC} is blind and the position of the trap qubits $t$ remains unknown to Bob. Recall that at each stage $i$ only qubit $i$ is measured. We present some intermediate definitions before formalizing the definition of verification. All the protocols presented so far describe the expected behavior of Alice and Bob in a hiding protocol. Since we are concerned with the secrecy of Alice's resources we can assume that Alice always follows the steps of the protocol. In fact after the initial step when Alice draws all the random variables $\theta_i$ and $r_i$ her behavior, for a fixed run of the protocol, is deterministic. This means that at each step the next move of Alice is determined completely by the past, however a malicious Bob might deviate in any way he desires. We will define a run of protocol to be \emph{honest} (Bob has behaved as expected) or \emph{correct} (the output is correct despite Bob's deviations) based on the outcome of all measurements and the quantum output state if it exists.

Recall that in a generic hiding protocol with quantum input and output the messages sent by Bob to Alice depend on a collection of outcome measurements, $s_i \in \{0,1\}$. In fact Bob will send the outcome value $b_i$ and then Alice, depending on $r_i$, will reset them to their corrected values $s_i$. In what follows we will deal with the corrected outcome measurement that is $s_i$. Similarly at the end of the protocol Bob will send Alice some quantum output state in the output Hilbert space $\mathcal{H}_O$ that needs to be corrected depending on all the measurements outcomes. In what follows we consider the corrected quantum output state $\rho$. Note that the values of $s_i$ and $\rho$ depends on Alice's specific random choices and also Bob's general strategy of deviation. We treat this information as a single density operator to deal uniformly with both classical and quantum output. Finally in order to consider the most general deviation that Bob can perform during a run of protocol we consider a collection of unitary operators acting each at a stage of the protocol on the private qubits of Bob and all the other qubits and classical bits sent by Alice to Bob. 

\begin{definition}\label{d-outcome}
Consider a particular run of a generic hiding protocol, where all the following parameters are fixed: Alice's angles of measurements $\phi=(\phi_i)_{1\leq i \leq (m-n)}$; Alice's random variables $x=(x_i)_{1\leq i \leq n}$, $r=(r_i)_{1\leq i \leq (m-n)}$, $\theta=(\theta_i)_{1\leq i \leq m}$ and $d = (d_i)_{i \in D}$; Alice's input state $\ket I$; The number of Bob's private qubits $B$; Bob's deviation unitaries at each stage of the protocol $\mathcal{U}=\{U_i\}_{0\leq i \leq m+1}$ acting on all quantum and classical messages. We denote the \emph{outcome density operator} (of all classical and quantum messages sent by Bob to Alice) as follows:   
\begin{align*}
\mathcal{B}_j(\nu) = \sum_{\vec{s}\in\{0,1\}^{|O|^c}} \; p_{\nu,j}(\vec{s}) \; \ket{\vec{s}} \bra{\vec{s}} \otimes \rho_{\nu,j}^{\vec{s}}
\end{align*}
where $\nu$ collectively denotes Alice's choice of variables $t,x,r,\theta, d$; $j$ ranges over Bob's choices: $B$ and $\mathcal{U}$; $\vec s$ ranges over all possible values of the corrected values $\{s_i\}$ of the measurement outcomes $\{b_i\}$ sent by Bob to Alice; and $\rho_{\nu,j}^{\vec{s}}$ is the reduced density operator for the non-measured qubits with the corresponding correction operators for the measurement outcomes $\vec s$ has been applied. We call the outcome density operator $\mathcal{B}_{0}(\nu)$, obtained from a run of the protocol where all $U_i$ are set to be the identity operator, the \emph{exact outcome density operator}. That is the outcome density operator obtained from a run where Bob exactly follows the step of the protocol. 
\end{definition}

Note that if we were dealing only with a deterministic pattern over a connected graph state then the outcome density operator could have been simplified to a fixed pure state of the output qubits, independent of the measurement outcomes. Moreover in such a scenario the probability of each branch of the computation would have been the same. However the above definition aims to capture any general deviation by Bob, that could affect the determinism and probability of the branches. Also since we will have dummy and trap qubits then not all the possible branches will be equally probable. The outcome density operator, depending on all the random choices of Alice and Bob, can be classified as follows below. Although not all mentioned categories will be used in the remainder of the paper, we give them here for completeness and to highlight the subtle differences between possible outcomes.

\begin{definition}
We say the outcome density operator $\mathcal{B}_j(\nu)$ is \emph{honest} if it is indistinguishable from the exact outcome density operator:
\begin{align*}
\lVert {\mathcal{B}_j(\nu) - \mathcal{B}_{0}(\nu)} \rVert_{tr}= 0,
\end{align*}
where $\lVert \cdot \rVert_{tr}$ denotes the trace norm.
It is called \emph{correct} if the quantum output state and the trap outcome measurement is indistinguishable from the corresponding value of the exact outcome density operator:
\begin{align*}
\lVert \mbox{\em Tr\em}_{i\not \in O, i\not = t} (B_j(\nu)) - \mbox{\em Tr\em}_{i\not \in O, i\not = t} (B_{0}(\nu)) \rVert_{tr} = 0.
\end{align*}
It is called \emph{lucky} if $b_t = r_t$ and finally it is called \emph{incorrect} if it is lucky but the quantum output state, $Tr_{i \not \in \{O \setminus \{t\}\}} (B_j(\nu))$, is orthogonal to the corresponding subsystem of the exact outcome density operator. Note that for the classical output scenario, any bit-flip implies orthogonality.
\end{definition}

Alice should not care if Bob's deviation leads to a correct outcome density operator, as the final quantum or classical output is in the correct state. Therefore, in the definition of a verifiable blind quantum computation we aim to bound the probability of Alice being fooled, i.e the probability of Alice accepting an incorrect outcome density operator. Any outcome density operator either results in $s_t \neq r_t$ or is contained within the subspace of correct and incorrect outcome states. Hence intuitively, a protocol is defined to be verifiable if the corresponding outcome state is \emph{far from} any incorrect outcome states. Following the approach of \cite{BCGST02}, we first define the notion of correctness. Recall that for simplicity we have assumed that the computation is deterministic and the input is in a pure state, and hence the ideal output will necessarily be a pure state. This restriction to pure states mirrors the approach of \cite{BCGST02}.

\begin{definition}\label{d-auth}
Let $P_\text{incorrect}^\nu$ be the projection onto the subspace of all the possible incorrect outcome density operator for the fixed choice of Alice's random variables $\nu$. It will be convenient to divide $\nu$ into two subsets depending on whether the secret variables correspond to the trap setting or the remainder of the computation. Thus we define $\nu_T = \{t,r_t,\theta_t\}$ and $\nu_C = \nu/\nu_T$. When the output state is a pure state, $P^{\nu}_\text{incorrect}$ is given by
\begin{align*}
 ( \mathbb{I} - \ket {\Psi_{ideal}} \bra{\Psi_{ideal}} ) \otimes \ket {\eta_t^{\nu_T}} \bra {\eta_t^{\nu_T}}
\end{align*}
where $\ket{\Psi_{ideal}} \bra{\Psi_{ideal}}  =  Tr_{i \not \in \{O \setminus \{t\}\}} (B_0(\nu))$, and where $\ket{\eta_t^{\nu_T}} = \ket{+_{\theta_t}}$ when $t\in O$ and $\ket{\eta_t^{\nu_T}} = \ket{r_t}$ otherwise. Let $p(\nu)$ be the probability of Alice choosing random variables parameterized by $\nu$, that is the probability of choosing a particular vertex, among all possible vertices of the graph, to be the trap position (denoted as a random variable $t$) and the probability of choosing random variables $r,x,\theta$ and $d$ (as defined in Definition \ref{d-outcome}). Given $0 \leq \epsilon < 1 $, we define a protocol to be \emph{$\epsilon$-verifiable}, if for any choice of Bob's strategy (defined as in Definition \ref{d-outcome} and denoted by index $j$) the probability of Alice accepting an incorrect outcome density operator is bounded by $\epsilon$:
\begin{align*}
\text{\em Tr\em} (\sum _\nu \; p(\nu)\; P_\text{incorrect}^\nu \;  B_j(\nu) ) \leq \epsilon. 
\end{align*}
\end{definition}
Recall that $B_0(\nu)$ is the output density operator of an honest run after the corrections have been performed. Hence, in the above definition $\ket{\Psi_{ideal}}$ is independent of $\nu$, since for an honest run of the protocol, the output state is independent of Alice's secret parameters, via the correctness theorem.

\begin{theorem} \label{t-authen1}
Protocol \ref{prot:DTBQC} is $(1 - \frac{1}{2m})$-verifiable in general, and in the special case of purely classical output the protocol is also $(1 - \frac{1}{m})$-verifiable, where $m$ is the total number of qubits in the protocol.
\end{theorem}
\begin{proof}
At the beginning of the protocol, Alice chooses the independent and uniform random variables for $\nu$. Next Alice prepares the input qubits in the following form:
\begin{align*}
\ket {e^\nu} = X_1^{x_1} Z_1(\theta_1) \otimes \ldots \otimes  X_n^{x_l} Z_n(\theta_l) \ket I
\end{align*}
and positions them among the first $n$ qubits. Recall that $n > |I|$ and hence the trap qubit might be among this set of qubits. She then prepares the remaining qubits in the following form (where $D$ is the index of the dummy qubits)
\AR{
\forall i\in D &\;\;\;& \ket {d_i} \\ 
\forall i \not \in D &\;\;\;& \prod_{j\in N_G(i) \cap D} Z^{d_j}\ket {+_{\theta_i}}~=\ket {+_{\theta_i + \sum_{j\in N_G(i) \cap D} d_j \pi}}
}
and sends all $m$ qubits in the order of the labeling of the vertices of the graph, we represent the whole $m$ qubit state as $\ket {M^\nu}$. We can treat all the measurement angles $\delta_i$ as orthogonal quantum states $\ket {\delta_i}$. For a fixed choice of Alice's random variables ($\nu$) and Bob's strategy ($j$), Bob's output from the computation can be written in the form of the output of a circuit computation as depicted in Figure \ref{fig:devfig1}. Note this is the state of the system before the relevant corrections for Alice's secret key have been applied to yield the outcome density operator $B_j(\nu)$.

\begin{figure}[h!]
\begin{center}
\includegraphics[width=\columnwidth]{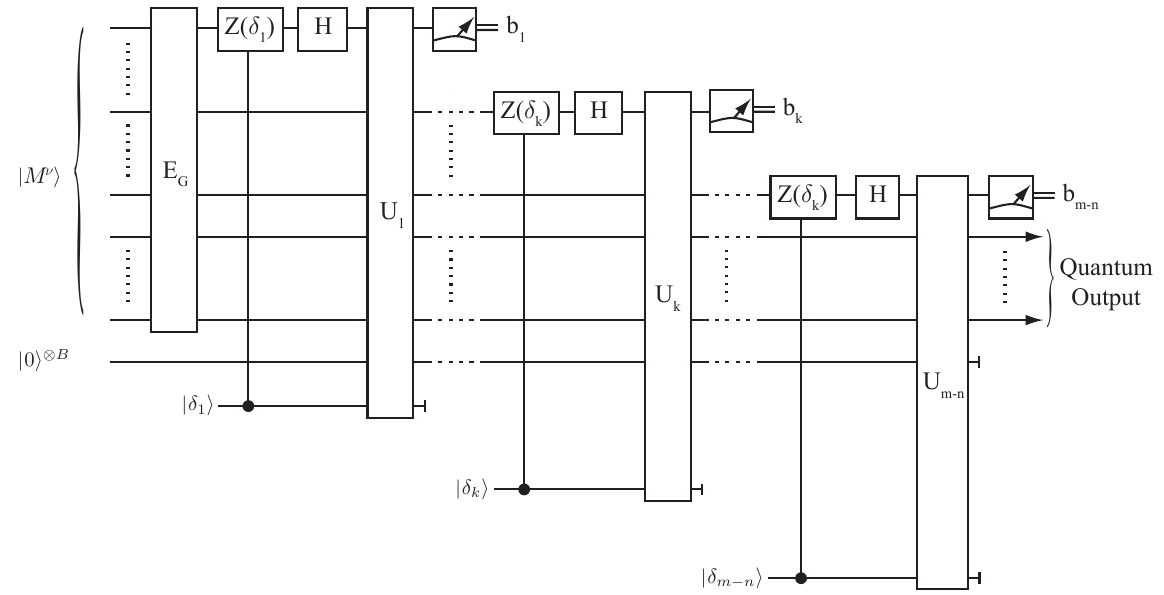}
\caption{A run of protocol together with Bob's deviation represented as $U_i$ operators. The entangling operator, $E_G$, is the collection of all the required \textsc{ctrl}-$Z$ operators corresponding to the graph edges. Note that in Definition \ref{d-outcome} we also considered an operator $U_0$ representing Bob's initial deviation. In the figure, for simplicity, we have commuted $U_0$ and combined it with $U_1$. Trivially, if all the $U_i$ operators are set to be identity the above circuit converges to the exact run of the protocol, where a measurement in the basis $\ket {\pm_{\delta_i}}$ is implemented using the controlled $Z$-rotation followed by a Hadamard gate and finally a Pauli $Z$ basis (computation basis) measurement on the corresponding qubits.}
\label{fig:devfig1}
\end{center}
\end{figure}

While in the actual protocol, at step $i$, Alice computes $\delta_i$ as a function of $s_{<i}$ which in turn is calculated from $b_{<i}$ and $r_{<i}$, we can rewrite the circuit from Figure \ref{fig:devfig1} in such a way that the values $\delta_i$ are part of the initial state, without affecting causality as they do not interact with anything until after the corresponding $b_i$ has been generated. This intuition is made rigorous in Equation \ref{e-consist} via the inclusion of projections to ensure consistency. This will allow us to reorder all the operators $U_i$ to the end to obtain the new circuit shown in Figure \ref{fig:devfig2}. Note that Figure \ref{fig:devfig2} is not an actual run of the protocol, it is a mathematical equivalent of Figure \ref{fig:devfig1} where the values of $b_i$ have been fixed to permit us to commute the operators as depicted. However in the following proof we have considered any general deviation performed by Bob, that is to say we consider any arbitrary $U_i$ operators. 

\begin{figure}[h!]
\begin{center}
\includegraphics[width=\columnwidth]{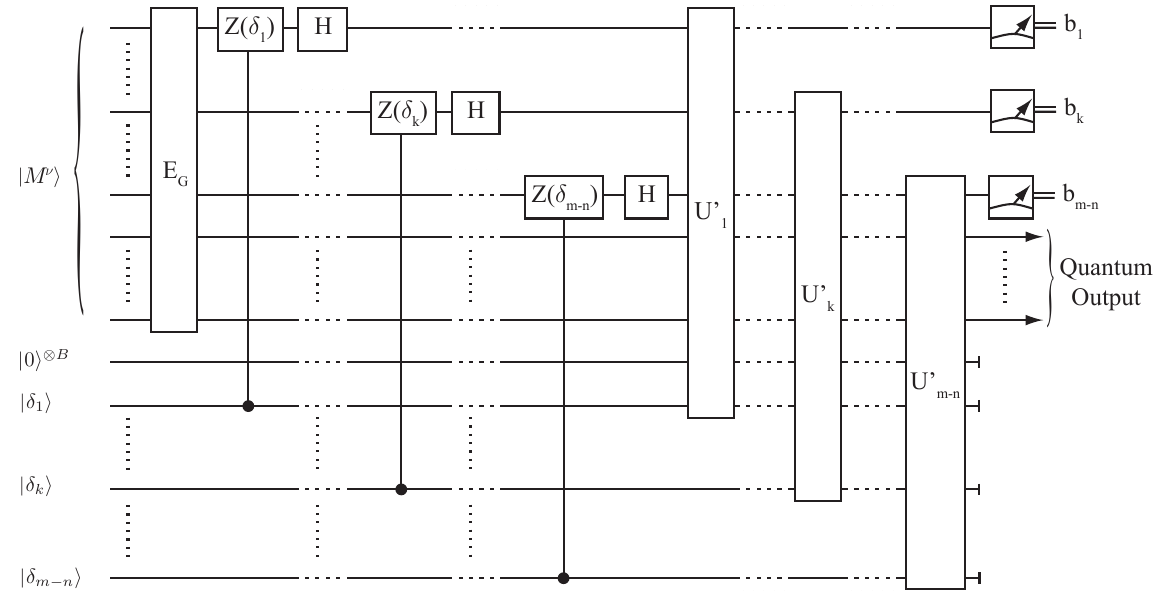}
\caption{The fact that any $U_j$ in Figure \ref{fig:devfig1} is independent of all $\delta_{i>j}$, allows us to reposition the deviation to the end of the circuit as shown above. Hence we can rewrite Bob's deviation as $U'_i = \mathcal{P}_i U_i \mathcal{P}_i^\dagger$, where $\mathcal{P}_i = \bigotimes_{i+1\leq j\leq m-n} H_j Z_j(\delta_j)$.}
\label{fig:devfig2}
\end{center}
\end{figure}

In the rest of this proof we will use $t$ to represent both the random variable and also the position of the trap qubit. We denote by $\Omega = U'_{m-n} U'_{m-n-1} ... U'_1$ the overall action of Bob's deviation and by $\mathcal{P} = \big(\bigotimes_{1\leq i\leq m-n} H_i Z_i(\delta_i)\big) E_G$ the action of the exact protocol prior to measurement. Here, and in Figure \ref{fig:devfig2}, we have taken $U'_i = \mathcal{P}_i U_i \mathcal{P}_i^\dagger$, where $\mathcal{P}_i = \bigotimes_{i+1\leq j\leq m-n} H_j Z_j(\delta_j)$. Further we denote by  
\AR{
\ket {\Psi^{\nu,b}} = \bigotimes_{1\leq i\leq m} \ket {M^{\nu}} \bigotimes_{1\leq j\leq m-n} |\delta_j^b
\rangle}
the joint state of the initial (input, dummy and prepared) qubits sent by Alice to Bob and the classical angles $\delta_i^{b}$, where $b$ represents a possible branch of the computation as parameterized by the measurement results $\{b_i\}$ sent by Bob to Alice. Finally, {in line with} Definition \ref{d-auth}, we define $C_{\nu_C,b}$ to be the Pauli operator which maps the final quantum output state to the correct one depending on the random variable $\nu_C$ and computation branch $b$. Hence we have 
\EQ{\label{e-consist}
B_j(\nu)={\mbox{Tr}_B}\left(\sum_{b} \ket{b+c_r} \bra{b} C_{\nu_C,b} \Omega \mathcal {P} ((\otimes^B \ket 0 \bra 0) \otimes |\Psi^{\nu,b}\rangle\langle \Psi^{\nu,b}|) \mathcal{P}^\dagger \Omega^\dagger C_{\nu_C,b}^\dagger \ket {b} \bra {b+c_r}\right).
}
where $(c_r)_i = r_i$ for all $i\neq t$ and $(c_r)_t = 0$, and the subscript $B$ denotes that the partial trace is taken over Bob's private register. Here $c_r$ is used to compactly deal with the fact that in the protocol all measured qubits are decrypted by XORing them with $r$, except for the trap qubit which remains uncorrected. Note that in the above the operator $\bra b \cdots \ket b$ acts upon the subspace of all measured qubits and $\ket {b+c_r} \cdots \bra {b+c_r}$ store the corrected outcome of the measurement. The above equation includes the dependence of $\delta_i$ on previous measurement results via the inclusion of the parameter $b$ in the initial state $|\Psi^{\nu,b}\rangle$. The projectors $\ket{b+c_r} \bra{b}$ and $\ket {b} \bra {b+c_r}$ then enforce consistency, by ensuring that measurement results match the values used in the computation of subsequent $\delta_i$.

We take $P_{\bot}$ to be the projection onto the subspace of incorrect states for the non-trap qubits, after Alice's final corrections have been applied to any quantum output. Hence 
\begin{align*}
P_\text{incorrect}^\nu = P_{\bot} \otimes \ket {\eta_t^{\nu_T}} \bra {\eta_t^{\nu_T}}
\end{align*}
where $\ket{\eta_t^{\nu_T}} = \ket{r_t}_t$ for $1 \leq t\leq m-n$ and $\ket{\eta_t^{\nu_T}} = \ket{+_{\theta_t}}_t$ for $m-n+1 \leq t \leq m$. Here we use the subscript on the ket to identify the relevant qubit.  Thus we have
\begin{align*}
\mbox{Tr}(P^{\nu}_\text{incorrect} \; B_j(\nu) ) =&\mbox{Tr} \Bigg( P_{\bot} \otimes \ket {\eta_t^{\nu_T}} \bra {\eta_t^{\nu_T}} \Bigg(\sum_{b} \ket{b+{c_r}}\bra{b} C_{\nu_C,b}\Omega \mathcal {P} \\
&~~~~~~~~\left(\left(\otimes^B \ket 0 \bra 0\right) \otimes |\Psi^{\nu,b}\rangle\langle \Psi^{\nu,b}|\right) \mathcal {P}^\dagger \Omega^\dagger C_{\nu_C,b}^\dagger \ket{b}\bra{b+c_r}\Bigg)\Bigg).
\end{align*}

As Bob's private register is traced out, the net result of $\Omega$ is to apply a completely positive trace preserving map of the other qubits. {Taking the Kraus operators associated with this operator to be $\{\chi_k\}$, with $\sum_k \chi_k \chi_k^\dagger = \mathbb{I}$, we have}
\begin{align*}
\mbox{Tr}(P^{\nu}_\text{incorrect} \; B_j(\nu) ) = {\sum_k}\sum_b \mbox{Tr} \Bigg(& \left(P_{\bot} \otimes \ket {\eta_t^{\nu_T}} \bra {\eta_t^{\nu_T}}\right) \ket{b+c_r}\bra{b} C_{\nu_C,b} \chi_k \mathcal {P} \\ & |\Psi^{\nu,b}\rangle\langle \Psi^{\nu,b}| \mathcal {P}^\dagger {\chi_k}^\dagger C_{\nu_C,b}^\dagger \ket{b}\bra{b+c_r}\Bigg).
\end{align*}
Since any {Kraus} operator can be written as a linear combination of Pauli operators with complex coefficients, we have ${\chi_k = \sum_{i} \alpha_{ki} \sigma_i}$, where ${\sum_k} \sum_i\alpha_{ki}\alpha_{ki}^* = 1$ and $\sigma_i$ is a Pauli operator acting on the joint quantum state of the system. Therefore the above equation can be written as 
\begin{align*}
\mbox{Tr} (P^{\nu}_\text{incorrect} B_j(\nu) ) &= {\sum_k} \sum_b \mbox{Tr} \Bigg( \left(P_{\bot} \otimes \ket {\eta_t^{\nu_T}} \bra {\eta_t^{\nu_T}}\right) \ket{b+c_r}\bra{b} C_{\nu_C,b} \\
&~~~~~~~~~~~~~~~~~~~~~~~~\left(\displaystyle\sum _{i,j} \alpha_{ki} \alpha_{kj}^*   \sigma_i \mathcal {P} \ket {\Psi^\nu} \bra {\Psi^\nu}  \mathcal {P}^\dagger \sigma_j\right) C_{\nu_C,b}^\dagger \ket{b}\bra{b+c_r} \Bigg) \\
&= {\sum_k} \sum_b \mbox{Tr} \Bigg(\displaystyle\sum _{i,j} \alpha_{ki} \alpha_{kj}^* \left( P_{\bot}  \otimes \ket {\eta_t^{\nu_T}} \bra {\eta_t^{\nu_T}} \right) \ket{b+c_r}\bra{b} C_{\nu_C,b}\\&~~~~~~~~~~~~~~~~~~~~~~~~ \sigma_i \mathcal {P} | \Psi^{\nu,b}\rangle \langle\Psi^{\nu,b} |\mathcal {P}^\dagger \sigma_j C_{\nu_C,b}^\dagger\ket{b}\bra{b+c_r}\Bigg).
\end{align*}

In order to determine which $\sigma_i$ terms have a non-zero contribution in the above sum after the projection operator is taken into account, it will be necessary to look at the structure of each such Pauli operator. To this end, we will denote by $\sigma_{i|\gamma}$ the action of $\sigma_i$ on qubit $\gamma$, and hence $\sigma_{i|\gamma} \in \{I,X,Y,Z\}$. For simplicity we assume each $\delta_i$ is encoded across 3 qubits (since there are only 8 possible angles). Thus, we have $1\leq \gamma\leq (m+ 3(m-n))$, where $1\leq \gamma \leq m$ identifies qubits received from Alice and the remaining $\gamma$ values identify the qubits containing $\delta_i$. Without loss of generality, we can assume that the qubits representing the values of $\delta$ remain unchanged by Bob's deviation, and hence we can take $\sigma_{i|\gamma} \in \{I,Z\}$ for all $m < \gamma$.

The probability of Alice accepting an incorrect outcome density operator is given by 
\AR{
p_\text{incorrect} = \mbox{Tr}(\sum _\nu \; p(\nu)\; P^{\nu}_\text{incorrect} \;  B_j(\nu) ) \, .
}
This can be calculated via the expression for $\mbox{Tr} ( P^{\nu}_\text{incorrect} \;  B_j(\nu))$ obtained earlier
\begin{align*}
p_\text{incorrect} &= \sum _\nu p(\nu)\mbox{Tr} (P^{\nu}_\text{incorrect} \;  B_j(\nu)) \\
&= \sum_{{k},b} \mbox{Tr} \bigg(\sum_\nu p(\nu) \displaystyle\sum _{i,j} \alpha_{ki} \alpha_{kj}^* \left( P_{\bot}  \otimes \ket {\eta_t^{\nu_T}} \bra {\eta_t^{\nu_T}} \right) \ket{b+c_r}\bra{b} \\
&~~~~~~~~~~~~~~~~~~~~~~~~~~~~~~~~~~~~~C_{\nu_C,b}\sigma_i \mathcal {P}  | \Psi^{\nu,b}\rangle \langle\Psi^{\nu,b} |\mathcal {P}^\dagger \sigma_j C_{\nu_C,b}^\dagger \ket{b}\bra{b+c_r} \bigg)\\
&= \sum_{b,i,j,{k}} \mbox{Tr} \bigg(\sum_\nu p(\nu) \alpha_{ki} \alpha_{kj}^* \left( P_{\bot}  \otimes \ket {\eta_t^{\nu_T}} \bra {\eta_t^{\nu_T}} \right) \ket{b+c_r}\bra{b}\\
&~~~~~~~~~~~~~~~~~~~~~~~~~~~~~~~~~~~~~ C_{\nu_C,b}\sigma_i \mathcal {P}  | \Psi^{\nu,b}\rangle \langle\Psi^{\nu,b} |\mathcal {P}^\dagger \sigma_j C_{\nu_C,b}^\dagger \ket{b}\bra{b+c_r}\bigg)
\end{align*}
By noting that $|b_j+c_{r_j}\rangle$ commutes with $| \Psi^{\nu,b}\rangle \langle\Psi^{\nu,b} |$ for all $j\neq t$, the above expression can be rewritten as
\begin{align*}
p_\text{incorrect} &= \sum_{b,i,j,{k}} \mbox{Tr} \bigg(\sum_\nu p(\nu) \alpha_{ki} \alpha_{kj}^* \left( P_{\bot}  \otimes \ket {\eta_t^{\nu_T}} \bra {\eta_t^{\nu_T}} \right) \\
&~~~~~~~~~~~~~~~~~~~~~~~~~~~~~~~~~~~~~ \ket{b_t}\bra{b} C_{\nu_C,b}\sigma_i \mathcal {P}  | \Psi^{\nu,b}\rangle \langle\Psi^{\nu,b} |\mathcal {P}^\dagger \sigma_j C_{\nu_C,b}^\dagger \ket{b}\bra{b_t}\bigg).
\end{align*}

In order to obtain an upper bound for the above expression we make use of sets of indices $\gamma$ of qubits such that the action of $\sigma_i$ at that position, $\sigma_{i|\gamma}$, is a particular Pauli operator, which we denote as follows:
\begin{align*}
A_i &= \{\gamma \;\;\; \mbox{s.t.} \;\;\;  \sigma_{i|\gamma} = I \mbox{ and } 1\leq \gamma \leq m\}\\
B_i &= \{\gamma \;\;\; \mbox{s.t.} \;\;\;  \sigma_{i|\gamma} = X \mbox{ and } 1\leq \gamma \leq m\}\\
C_i &= \{\gamma \;\;\; \mbox{s.t.} \;\;\;  \sigma_{i|\gamma} = Y \mbox{ and } 1\leq \gamma \leq m\}\\
D_i &= \{\gamma \;\;\; \mbox{s.t.} \;\;\;  \sigma_{i|\gamma} = Z \mbox{ and } 1\leq \gamma \leq m\}.
\end{align*}
Note that in the above we restrict attention to the set of qubits originally sent from Alice to Bob (which is why $1 \leq \gamma \leq m$), and disregard the action on Bob's private qubits. Additionally, we will make use of a superscript $O$ to denote subsets of the above sets subject to the constraint that $\gamma$ is an output qubit ($m-n<\gamma$). Thus, for example, $D^O_i = \{\gamma \;\;\; \mbox{s.t.} \;\;\;  \sigma_{i|\gamma} = Z \mbox{ and } m-n+1\leq \gamma \leq m\}$. 
We note that only $\sigma_i$ and $\sigma_j$ operators for which $\mbox{Tr}(P_{\bot}\sigma_i \mathcal{P}| \Psi^{\nu,b}\rangle \langle\Psi^{\nu,b} | \mathcal{P}^\dagger \sigma_j) \neq 0$ contribute to $p_\text{incorrect}$. With the above definitions in place, we can express succinctly a necessary condition for this to hold as $|B_i| + |C_i| + |D^O_i| \geq 1$ (denoted as $i\in E_i$) and $|B_j| + |C_j| + |D^O_j| \geq 1$ (denoted as $j\in E_j$). That is to say, one or both of the following has happened: $\sigma_i$ ($\sigma_j$) has produced an incorrect outcome for one or more of the measurement results and hence $|B_i\setminus B_i^O| + |C_i \setminus C_i^O| \geq 1$ ($|B_j\setminus B_j^O| + |C_j\setminus C_j^O| \geq 1$) or $\sigma_i$ ($\sigma_j$) acts non-trivially on the quantum output and hence $|B_i^O| + |C_i^O| + |D_i^O| \geq 1$ ($|B_i^O| + |C_j^O| + |D_j^O| \geq 1$). Using this set notion and by taking the trace over the subspace of the measurement results except for the trap qubit we obtain
\begin{align*}
p_\text{incorrect} =& \sum_{{k},b} \sum_{i\in E_i} \sum_{j \in E_j} \mbox{Tr} \Bigg(\sum _{\nu} p(\nu)  \alpha_{ki} \alpha_{kj}^* \left( P_{\bot}  \otimes \ket {\eta_t^{\nu_T}} \bra{\eta_t^{\nu_T}}\right)\\
&~~~~~~~~~~~~~~~~~~~~~~~~~~~~~~~~~~\ket{b_t}\bra{b}C_{\nu_{C},b}\sigma_i \mathcal {P}  | \Psi^{\nu,b}\rangle \langle\Psi^{\nu,b} |\mathcal {P}^\dagger \sigma_j C_{\nu_{C},b}^\dagger\ket{b}\bra{b_t}\Bigg),
\end{align*}
where we take $\ket{b_t}$ to have have unit dimension if $t\in O$. The reason for doing this is to allow a uniform treatment of trap qubits independent of whether or not the trap occurs on a measured qubit. Taking $b' = \{b_i\}_{i\neq t}$, a substring of $b$ which excludes the value for the trap measurement, the above equation can be written as
\begin{align*}
p_\text{incorrect} =&  \sum_{{k},b} \sum_{i\in E_i} \sum_{j \in E_j} \mbox{Tr} \Bigg(\sum _{\nu} p(\nu)  \alpha_{ki} \alpha_{kj}^* \left( P_{\bot}  \otimes \left(\ket {\eta_t^{\nu_T}} \bra{\eta_t^{\nu_T}}b_t\rangle\bra{b_t}\right)\right)\\
&~~~~~~~~~~~~~~~~~~~~~~~~~~~~~\bra{b'}C_{\nu_{C},b}\sigma_i \mathcal {P}  | \Psi^{\nu,b}\rangle \langle\Psi^{\nu,b} | \mathcal {P}^\dagger \sigma_j C_{\nu_{C},b}^\dagger \ket{b'}\Bigg)\\
\end{align*}
Note in the above that if the trap is measured we have $\bra{\eta_t^{\nu_T}} {b_t}\rangle = \delta_{\eta_t^{\nu_T},b_t}$, otherwise $\ket{b_t}\bra{b_t} = 1$. Hence we have

\begin{align*}
p_\text{incorrect} =& \sum_{{k},b'} \sum_{i\in E_i} \sum_{j \in E_j} \mbox{Tr} \Bigg(\sum _{\nu} p(\nu)  \alpha_{ki} \alpha_{kj}^* \left( P_{\bot}  \otimes \ket {\eta_t^{\nu_T}} \bra {\eta_t^{\nu_T}}\right)\ket{b'}\bra{b'}\\
&~~~~~~~~~~~~~~~~~~~~~~~~~~~~~C_{\nu_{C},b'}\sigma_i \mathcal {P}  | \Psi^{\nu,b'}\rangle \langle\Psi^{\nu,b'} | \mathcal {P}^\dagger \sigma_j C_{\nu_{C},b'}^\dagger \Bigg)\\
=& \sum_{{k},b'} \sum _{\nu} p(\nu) \mbox{Tr} \Bigg(\left( P_{\bot}  \otimes \ket {\eta_t^{\nu_T}} \bra {\eta_t^{\nu_T}}\right) \ket{b'}\bra{b'}\\
&~~~~~~~~~~~~~~~~~~~~~~~~~~~~~C_{\nu_{C},b'}\left( \sum_{i\in E_i} \alpha_{ki} \sigma_i\right) \mathcal {P}  | \Psi^{\nu,b'}\rangle \langle\Psi^{\nu,b'} | \mathcal {P}^\dagger \left( \sum_{i\in E_i} \alpha_{ki} \sigma_i\right)^\dagger C_{\nu_{C},b'}^\dagger \Bigg)\\
\leq& \sum_{{k},b'} \sum _{\nu} p(\nu) \mbox{Tr} \Bigg(\left(\ket {\eta_t^{\nu_T}} \bra {\eta_t^{\nu_T}}\otimes \ket{b'}\bra{b'}\right) \\ 
&~~~~~~~~~~~~~~~~~~~~~~~~~~~~~C_{\nu_{C},b'}\left( \sum_{i\in E_i} \alpha_{ki} \sigma_i\right) \mathcal {P}  |\Psi^{\nu,b'}\rangle\langle\Psi^{\nu,b'}|\mathcal {P}^\dagger \left( \sum_{i\in E_i} \alpha_{ki} \sigma_i\right)^\dagger C_{\nu_{C},b'}^\dagger \Bigg)\\
=& \sum_{{k},b'} \sum _{\nu} p(\nu) \mbox{Tr} \Bigg(\left(\ket {\eta_t^{\nu_T}} \bra {\eta_t^{\nu_T}}\otimes \ket{b'}\bra{b'}\right) \left( \sum_{i\in E_i} \alpha_{ki} \sigma_i\right) \mathcal {P}  |\Psi^{\nu,b'}\rangle\langle \Psi^{\nu,b'}| \mathcal {P}^\dagger \left( \sum_{i\in E_i} \alpha_{ki} \sigma_i\right)^\dagger \Bigg),
\end{align*}
where the inequality follows from the fact that the projector, $P_\bot$, acts on a positive semi-definite matrix, and the last equality follows from the fact that both remaining projectors act as the identity on qubits in $O$.

Next, we attempt to show that a necessary requirement for a term in the above summation over $i$ and $j$ to be non-zero is that $i=j$. As per the proof of blindness, summing over $\nu_C$ yields the maximally mixed state of the system received from Alice. Hence we have
\begin{align*}
p_\text{incorrect} \leq& \sum_{k,b',\nu_T} \sum_{i \in E_i} \sum_{j \in E_j} \alpha_{ik} \alpha_{jk}^* p(\nu_T) \mbox{Tr} \Bigg(\left(\ket {\eta_t^{\nu_T}} \bra {\eta_t^{\nu_T}}\otimes \ket{b'}\bra{b'}\right) \\
&~~~~~~~~~~~~~~~~~~~~~~~~~~~~~\sigma_i \left(\ket {\eta_t^{\nu_T}} \bra {\eta_t^{\nu_T}} \otimes \ket{\delta_t}\bra{\delta_t} \otimes {\frac{I}{\text{Tr}(I)}} \right) \sigma_j \Bigg)\\
=& \sum_{k,\nu_T} \sum_{i \in E_i} \sum_{j \in E_j} \alpha_{ik} \alpha_{jk}^* p(\nu_T) \mbox{Tr} \Bigg(\ket {\eta_t^{\nu_T}} \bra {\eta_t^{\nu_T}}\sigma_i \left(\ket {\eta_t^{\nu_T}} \bra {\eta_t^{\nu_T}} \otimes \ket{\delta_t}\bra{\delta_t} \otimes {\frac{I}{\text{Tr}(I)}} \right) \sigma_j \Bigg)\\
=& \sum_{k,\nu_T} \sum_{i \in E_i} \sum_{j \in E_j} \alpha_{ik} \alpha_{jk}^* p(\nu_T) \mbox{Tr} \Bigg(\bra {\eta_t^{\nu_T}}\sigma_i \left(\ket {\eta_t^{\nu_T}} \bra {\eta_t^{\nu_T}} \otimes \ket{\delta_t}\bra{\delta_t} \otimes {\frac{I}{\text{Tr}(I)}} \right) \sigma_j  \ket {\eta_t^{\nu_T}}\Bigg).
\end{align*}
As all Pauli matrices other than the identity are traceless, any terms in the sum which are non-zero necessarily have $\sigma_{i|\gamma} = \sigma_{j|\gamma}$ everywhere except for $\gamma=t$ and the corresponding delta register. We then consider the two cases corresponding to whether the trap is located in the quantum output or not separately. If $t\in O$ then the delta register does not exist, and using the fact that $\sum_{\theta_t,r_t}\mbox{Tr} \big(\bra {\eta_t^{\nu_T}} \sigma_i \ket {\eta_t^{\nu_T}} \bra {\eta_t^{\nu_T}}  \sigma_j \ket {\eta_t^{\nu_T}}\big)=0$, unless $\sigma_{i|t} = \sigma_{j|t}$, we arrive at the conclusion that the only terms which contribute to $p_\text{incorrect}$ are those where $\sigma_i = \sigma_j$. If, on the other hand, $t\notin O$, then averaging over $r_t$ alone is sufficient to give $\mbox{Tr} \big(\bra {\eta_t^{\nu_T}} \sigma_i \ket {\eta_t^{\nu_T}} \bra {\eta_t^{\nu_T}}  \sigma_j \ket {\eta_t^{\nu_T}}\big) = 0$, and hence $\sigma_{i|t} = \sigma_{j|t}$. In this case, averaging over $\theta$ yields the $\delta_t$ register in the maximally mixed state, and hence as before $\sigma_i$ and $\sigma_j$ must act identically on these qubits too, in order to avoid contributing zero to the value of $p_\text{incorrect}$. Consequently the only terms which contribute are those for which $\sigma_i = \sigma_j$. Using this identity with our previous expression for $p_\text{incorrect}$, we obtain
\begin{align*}
p_\text{incorrect} \leq& \sum_{k,\nu_T} \sum_{i \in E_i} \alpha_{ik} \alpha_{ik}^* p(\nu_T) \mbox{Tr} \Bigg(\bra {\eta_t^{\nu_T}}\sigma_i \left(\ket {\eta_t^{\nu_T}} \bra {\eta_t^{\nu_T}} \otimes \ket{\delta_t}\bra{\delta_t} \otimes {\frac{I}{\text{Tr}(I)}} \right) \sigma_i  \ket {\eta_t^{\nu_T}}\Bigg)\\
=& \sum_{k,\nu_T} \sum_{i \in E_i} |\alpha_{ik}|^2 p(\nu_T) \mbox{Tr} \left(\bra {\eta_t^{\nu_T}}\sigma_{i|t} \ket {\eta_t^{\nu_T}} \bra {\eta_t^{\nu_T}}  \sigma_{i|t}  \ket {\eta_t^{\nu_T}}\right) \\
=& \sum_{k,\nu_T} \sum_{i \in E_i} |\alpha_{ik}|^2 p(\nu_T) \left(\bra {\eta_t^{\nu_T}}  \sigma_{i|t} \ket {\eta_t^{\nu_T}}\right)^2\\
=&  \frac{1}{16m}{\sum_k} \sum_{i\in E_i} |\alpha_{ki}|^2  \sum_{t,r_t,\theta_t}  \left(\bra {\eta_t^{\nu_T}} \sigma_{i|t} \ket{\eta_t^{\nu_T}}\right)^2\\
=& \frac{1}{16m} {\sum_k} \displaystyle\sum _{i\in E_i} |\alpha_{ki}|^2 \Bigg(\displaystyle\sum_{t\leq m-n, \theta_t, r_t}  \big( \bra {\eta_t^{\nu_T}} \sigma_{i|t} \ket {\eta_t^{\nu_T}}\big)^2 + \displaystyle\sum_{m-n< t, \theta_t, r_t}  \big( \bra {\eta_t^{\nu_T}} \sigma_{i|t} \ket {\eta_t^{\nu_T}}\big)^2\Bigg) \\
=& \frac{1}{16m} {\sum_k} \displaystyle\sum _{i\in E_i} |\alpha_{ki}|^2 \left(\displaystyle\sum_{t\leq m-n, \theta_t, r_t}  \big(\bra {r_t} \sigma_{i|t} \ket {r_t} \big)^2  + \displaystyle\sum_{m-n< t, \theta_t, r_t}  \big( \bra {+_{\theta_t}} \sigma_{i|t} \ket {+_{\theta_t}}\big)^2\right)\\
=& \frac{1}{16m} {\sum_k} \displaystyle\sum _{i\in E_i} |\alpha_{ki}|^2 \left(\left(16|A_i\setminus A_i^O| + 16|D_i\setminus D_i^O|\right) +\left(8|B_i^O| + 8|C_i^O| + 16|A_i^O|\right)\right)\\
=& \frac{1}{2m} {\sum_k}\displaystyle\sum _{i\in E_i} |\alpha_{ki}|^2 \left(2|A_i| + 2|D_i\setminus D_i^O| + |B_i^O| + |C_i^O |\right).
\end{align*}

This can be further simplified, since $|A_i|+|B_i|+|C_i|+|D_i| = m$, giving
\begin{align*}
p_\text{incorrect} &\leq \frac{1}{2m} {\sum_k} \displaystyle\sum _{i\in E_i} |\alpha_{ki}|^2 \left( 2m - 2(|B_i| + |C_i| + |D_i^O|) + |B_i^O| + |C_i^O | \right)\\
&\leq \frac{1}{2m} {\sum_k}\displaystyle\sum _{i\in E_i} |\alpha_{ki}|^2 \left(2m - |B_i| - |C_i| - 2|D_i^O|\right)\\
&\leq \frac{1}{2m} {\sum_k} \displaystyle\sum _{i\in E_i} |\alpha_{ki}|^2 \left(2m - 1\right)\\
&\leq 1 - \frac{1}{2m}
\end{align*}
for the general case. However, for the specific case of only classical output, this bound can be made tighter by performing the simplification in a different way, since $|B_i^O|=|C_i^O|=|D_i^O|=0$, and hence
\begin{align*}
p_\text{incorrect} &\leq \frac{1}{2m} {\sum_k}\displaystyle\sum _{i\in E_i} |\alpha_{ki}|^2 \left(2|A_i| + 2|D_i\setminus D_i^O| + |B_i^O| - |C_i^O |\right)\\
&=\frac{1}{m} {\sum_k}\displaystyle\sum _{i:|B_i|+|C_i| \geq 1} |\alpha_{ki}|^2 \left(|A_i| + |D_i|\right) \\
&= \frac{1}{m} {\sum_k}\displaystyle\sum _{i:|B_i|+|C_i| \geq 1} |\alpha_{ki}|^2 \left(m - |B_i|-|C_i|\right) \\
&\leq \frac{1}{m} {\sum_k}\displaystyle\sum _{i:|B_i|+|C_i|\geq 1} |\alpha_{ki}|^2 \left(m - 1\right)\\
&\leq 1 - \frac{1}{m}.
\end{align*}\end{proof}

This single trap construction will be generalised in the next section to allow for exponential supression of the probability of accepting an incorrect outcome even in the case of quantum output. We finish this section by showing that even this simple construction can be used to verify universal quantum computation, using the cylinder brickwork state presented in Section \ref{s-resource}.

It is easy to verify that if Alice chooses a random row of a cylinder graph $\mathcal{G}^C_{n \times m}$ (Figure \ref{fig:cylinder}) and prepares all the qubits of that row in the states $\ket{z_i}$ where $z_i \in_R \{0,1\}$ and the rest of nodes in the state $\ket +$ then after entangling according to the cylinder brickwork graph the obtained state is a $\mathcal{G}_{(n-1) \times m} \bigotimes_{i=1}^m \ket {z_i}$. By choosing a random trap location and a dummy tape which contains its neighbourhood we can construct a single-trap verifiable universal blind quantum computing protocol, given by Protocol \ref{prot:Cylinder-DTBQC} and illustrated in Figure \ref{fig:Cylinder-DTBQC}.

\begin{figure}[h]
\begin{center}
\scalebox{0.5}[0.5]{\includegraphics{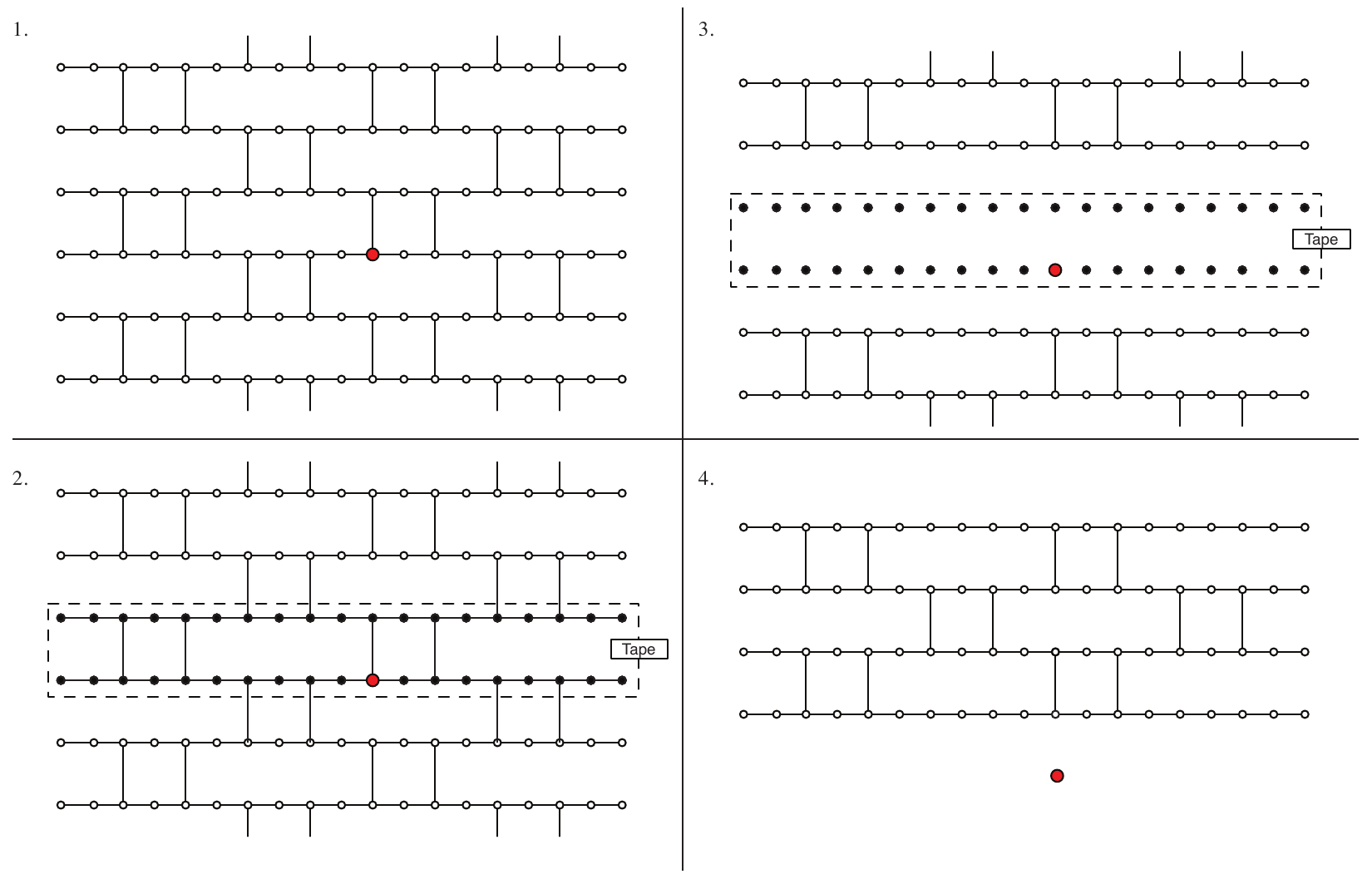}}
\caption{Single-trap verifiable universal blind quantum computation using the cylinder brickwork state: 1. A random qubit is chosen to be the trap qubit, the (red) filled node. 2. All other vertices in the tape containing the trap qubit the solid (black) nodes, are set to be dummy qubits. 3. This results in an isolated trap qubit in the state $\ket{+_{\theta}}$ together with many dummy qubtis after entaglement operations are applied by the server. 4. The net result, after discarding the dummy qubits, is a disentangled trap qubit in a product state with a brickwork state.}
\label{fig:Cylinder-DTBQC}
\end{center}
\end{figure}

\begin{algorithm}[h]
\caption{Single-Trap Verifiable Universal Blind Quantum Computation}
 \label{prot:Cylinder-DTBQC}
 \vskip 0.2 cm
\begin{itemize}
\item \textbf{Alice's resources} \\
\noindent -- A graph $G = \mathcal{G}^C_{n\times m}$ and a randomly chosen vertex $t$ of $G$. \\  
\noindent -- The rest of the resources are the same as Protocol \ref{prot:DBQC} where $\phi_{i} = 0$ for $i=t$ and $i\in D$ where $D$ contains the set of all vertices in a tape $T$ that contains position $t$ and all of its neighbours. 

\item \textbf{Follow the steps of Protocol \ref{prot:DBQC}.}

\item \textbf{Accept/Reject} \\
\noindent --  After obtaining all the output qubits from Bob, if the trap qubit, $t$, is an output qubit, Alice measures it with angle $\delta_{t} = \theta_t + r_t \pi$ to obtain $b_t$. \\
\noindent -- Alice accepts if $b_{t} = r_{t}$. 

\end{itemize}
\end{algorithm}

\begin{corollary}\label{t-Cylinder-DTBQC} Protocol \ref{prot:Cylinder-DTBQC} is universal, blind while leaking at most $m$ and $n$ as well as being is $(1 - \frac{1}{2m})$-verifiable in general and $(1 - \frac{1}{m})$-verifiable in the case of classical output. 
\end{corollary}
\begin{proof} Since the dummy qubits are prepared in eigenstates of Pauli $Z$ operator, they remain in a product state with the rest of the system after the entangling operations are applied by Bob. The result, as depicted in Figure \ref{fig:Cylinder-DTBQC}, is that the trap qubit also remains in a product state, and a brickwork state is prepared in the subsystem excluding $T$. The universality property then follows directly from the universality of the brickwork state from Theorem \ref{thm:universal}. As Protocol \ref{prot:Cylinder-DTBQC} is a special case of Protocol \ref{prot:DBQC}, the blindness property follows directly from Theorem \ref{t-DBQC} and therefore the angles of measurement $\phi_{i}$ remain secret from Bob. Moreover, the universality of the cylinder brickwork state guarantees that Bob's knowledge of $\mathcal{G}^C_{n\times m}$ does not reveal anything about the underlying computation except $n$ and $m$. As Protocol \ref{prot:Cylinder-DTBQC} is also a special case of Protocol \ref{prot:DTBQC}, the verifiability property follows directly from Theorem \ref{t-authen1}.
\end{proof}

\section{Probability Amplification for Universal Verifiable Blind QC}\label{s-amplification}

In the previous section we presented a very simple verifiable protocol where the probability of Bob succeeding in making Alice accept an incorrect outcome density operator was strictly less than 1. Building upon that simple construction, by adding more traps and making the computation fault tolerant, we can make the probability of Alice accepting an incorrect outcome density operator as small as required. The central idea is to design a protocol with $O(N)$ many traps in essentially random locations, where $N$ is the number of qubits in the protocol, to increase the probability of any local error being detected. The fault-tolerance is added to increase the minimum weight of any operator which leads to an incorrect outcome, and hence further increase the probability of detection. Here, and in what follows, the weight of a Pauli operator is defined to be the number of qubits upon which it acts non-trivially. First, given such a protocol we show how it amplifies the verification parameter. We then present the central contribution of this paper, a new universal verifiable blind quantum computing protocol that achieves the probability amplification without any such assumptions.

\begin{theorem} \label{t-authenT}
Let $\mathcal P$ be a blind quantum computing protocol on $N$ qubits with $N_T$ isolated traps in the states $\ket {+_{\theta_t}}$ at a set of positions $T$ chosen uniformly at random. Let $N_T/N$ be a constant $c$ and assume that the computation is encoded in such a way that any Pauli error with weight less than $d$ will be corrected or an error will be detected. Then the protocol is $(1 - \frac{c}{2})^d$-verifiable in general, and $(1 - c)^d$-verifiable in the case of purely classical output.
\end{theorem}

\begin{proof}
In order to exploit Theorem \ref{t-authen1}, we notionally partition the qubits into independent sets with one single trap qubit in each set. These partitions amount to extra information about the location of the trap qubits, and hence their inclusion can only serve to increase the probability of Bob convincing Alice to accept an incorrect state. Thus the bound we obtain with this additional information is still an upper bound on the probability of Alice accepting an incorrect output when these partitions are unknown. There are $N_T$ many such sets $S_\gamma$ with $1/c$ many qubits in each set. We adopt a similar proof strategy to that used to prove Theorem \ref{t-authen1}, taking
\begin{equation*}
P^{\nu}_\text{incorrect} = P_{\bot} \bigotimes_{t\in T} \ket {\eta_t^{\nu_T}} \bra {\eta_t^{\nu_T}}
\end{equation*}
as the projection onto the subspace of incorrect outcomes. As in the proof of Theorem \ref{t-authen1}, only those Pauli operators contribute to $p_\text{incorrect}$ where one or both of the following has happened: $\sigma_i$ has produced an incorrect outcome for some of the measurement results $b_i$ or $\sigma_i$ acts non-trivially on the quantum output. Now due to the error-detection property of the encoding assumed in the statement of the theorem we need to consider only those $\sigma_i$ where $|B_{i}|+|C_{i}|+|D_{i}^O| \geq d$. Following the steps of the proof of Theorem \ref{t-authen1} we obtain
\begin{align*}
p_\text{incorrect} &= \sum _\nu p(\nu) \mbox{Tr} (P^{\nu}_\text{incorrect} B_j(\nu) ) \\
&\leq {\sum_k} \displaystyle\sum _{i:|B_{i}|+|C_{i}|+|D_{i}^O| \geq d} |\alpha_{ki}|^2 \displaystyle\sum_{T}  p(T) \prod_{t\in T} \left( \sum_{\theta_t, r_t}p(\theta_t) p(r_t) \left(\bra {\eta_t^{\nu_T}} \sigma_{i|t} \ket{\eta_t^{\nu_T}}\right)^2 \right).
\end{align*}
Here we can exploit the structure we have introduced through the sets $S_\gamma$
\begin{align*}
p_\text{incorrect} \leq {\sum_k} \displaystyle\sum _{i:|B_{i}|+|C_{i}|+|D_{i}^O| \geq d} |\alpha_{ki}|^2 \displaystyle \prod_{\gamma=1}^{N_T}\sum_{t_\gamma, \theta_{t_\gamma}, r_{t_\gamma}}  p({t_\gamma})p(\theta_{t_\gamma})p(r_{t_\gamma}) \langle \eta^\nu_{t_\gamma} |\sigma_{i|{t_\gamma}} |\eta^\nu_{t_\gamma}\rangle^2.
\end{align*}
where $t_\gamma$ is taken to be the location of the trap qubit in set $S_\gamma$. Rearranging the above and substituting in the values of $p({t_\gamma})$, $p(\theta_{t_\gamma})$, and $p(r_{t_\gamma})$ we obtain
\begin{align*}
p_\text{incorrect} \leq {\sum_k} \displaystyle\sum _{i:|B_{i}|+|C_{i}|+|D_{i}^O| \geq d} |\alpha_{ki}|^2 \prod_{\gamma=1}^{N_T} \displaystyle\sum_{t_\gamma, \theta_{t_\gamma}, r_{t_\gamma}} \frac{c}{16} \langle \eta^\nu_{t_\gamma} |\sigma_{i|{t_\gamma}} |\eta^\nu_{t_\gamma}\rangle ^2.
\end{align*}
Note that within each set the position of the trap is chosen uniformly at random and so the probability of detection by that trap corresponds to the bound obtained for Theorem \ref{t-authen1}. Going through the steps of the proof of Theorem \ref{t-authen1} we obtain
\begin{align*}
p_\text{incorrect}  &\leq {\sum_k}  \displaystyle\sum _{i:|B_{i}|+|C_{i}|+|D_{i}^O| \geq d} |\alpha_{ki}|^2 \displaystyle\prod_{\gamma=1}^{N_T} \frac{c}{2}\big(2|A_{i\gamma}| + 2|D_{i\gamma}\setminus D_{i\gamma}^O| + |B_{i\gamma}^O|+|C_{i\gamma}^O|\big)\\
&= {\sum_k}  \displaystyle\sum _{i:|B_{i}|+|C_{i}|+|D_{i}^O| \geq d} |\alpha_{ki}|^2\displaystyle\prod_{\gamma=1}^{N_T} \frac{c}{2}\bigg(\frac{2}{c} - 2|D_{i\gamma}^O| -|B_{i\gamma}|-|C_{i\gamma}|-|B_{i\gamma}\setminus B_{i\gamma}^O| -|C_{i\gamma}\setminus C_{i\gamma}^O|\bigg),
\end{align*}
where we use the additional $\gamma$ subscript on sets $|A_{i\gamma}|, ...,|D_{i\gamma}|$ to indicate subsets of the respective sets, subject to the restriction that the elements are also in $S_\gamma$. For convenience we define $w_{i\gamma} = |B_{i\gamma}| + |C_{i\gamma}| + |D_{i\gamma}^O|$ and $w_i = |B_i| + |C_i| + |D_i^O|$. Thus we obtain
\begin{align*}
p_\text{incorrect} &\leq {\sum_k}  \displaystyle\sum _{i : w_{i} \geq d} |\alpha_{ki}|^2\displaystyle\prod_{\gamma=1}^{N_T} \frac{c}{2} \left(\frac{2}{c} - w_{i\gamma} -|B_{i\gamma}\setminus B_{i\gamma}^O|-|C_{i\gamma}\setminus C_{i\gamma}^O|-|D_{i\gamma}^O|\right)\\
&\leq {\sum_k} \displaystyle\sum _{i : w_{i} \geq d} |\alpha_{ki}|^2\displaystyle\prod_{\gamma=1}^{N_T}  \left(1 - \frac{c w_{i\gamma}}{2}\right).
\end{align*}
We now make use of the fact that, for any positive $a$, $1-\frac{ac}{2} \leq (1-(a-1)\frac{c}{2})(1-\frac{c}{2})$. As $w_{i\gamma}$ is a non-negative integer, we can recursively apply this identity to obtain
\begin{align*}
p_\text{incorrect} &\leq {\sum_k} \displaystyle\sum _{i : w_{i} \geq d} |\alpha_{ki}|^2 \displaystyle\prod_{\gamma=1}^{N_T} \big(1 - \frac{c}{2}\big)^{w_{i\gamma}}\\
&= {\sum_k} \displaystyle\sum _{i : w_{i} \geq d} |\alpha_{ki}|^2 (1 - \frac{c}{2})^{\sum_{\gamma=1}^{N_T} w_{i\gamma}}\\
&= {\sum_k} \displaystyle\sum _{i : w_{i} \geq d} |\alpha_{ki}|^2 (1 - \frac{c}{2})^{w_i}\\
&\leq {\sum_k}  \displaystyle\sum _{i : w_{i} \geq d} |\alpha_{ki}|^2 (1 - \frac{c}{2})^{d}\\
&\leq (1 - \frac{c}{2})^{d}.
\end{align*}
In the case of purely classical output this bound can be improved, since $|B_i^O|=|C_i^O|=|D_i^O|=0$. Going through the same steps with this additional constraint gives
\begin{align*}
p_\text{incorrect} &\leq {\sum_k} \displaystyle\sum _{i : w_{i} \geq d} |\alpha_{ki}|^2\displaystyle\prod_{\gamma=1}^{N_T}  \left(1 - {c w_{i\gamma}}\right)\\
&\leq (1 - c)^{d}.\end{align*}\end{proof}

We can now present the final contribution of this paper, a new scheme for blind quantum computing which has all the previously described properties: correctness, universality, blindness of angles, input, output and computation and more importantly verifiability with exponentially small probability of error. Roughly speaking, universality and correctness  will be obtained by using dotted-complete graph states (similar to Protocol \ref{prot:UBQC2}). In order to achieve verification we exploit the idea of dummy qubits (similar to Protocol \ref{prot:DBQC}) to create, blindly, out of a dotted-complete graph state $\tilde{\mathcal K}_{3N}$ three disconnected smaller dotted-complete graph states $\tilde{\mathcal K}_{N}$. Then we use two of these graph states to create $O(N)$ isolated trap qubits at random positions (similar to Protocol \ref{prot:DTBQC}). The final step is to perform the actual computation over the remaining dotted-complete graph state in such a way that the stated property in Theorem \ref{t-authenT} is also satisfied. That is, to have the measurement pattern encoded in such a way that any Pauli error with weight less than $d$, will be either corrected or detected. Such an encoding exists through the fault tolerant one-way quantum computing scheme of \cite{RHG07}. All that is needed is to create a three dimensional cluster state from the dotted-complete graph state and proceed with the fault tolerant computation scheme of Raussendorf, Harrington and Goyal \cite{RHG06,RHG07}{\footnote{{In its original form, this scheme requires $Z$-basis measurements to be made adaptively, which is not easily implementable using dummy qubits. However, the location of the dummy qubits can be fixed by always including a correction step for each gate teleportation in the logical circuit, where the angle of the correction is adapted based on the outcome of the teleportation. An alternative option is to use a slightly modified version of the scheme due to Morimae and Fujii \cite{FM12}, which requires only measurements in the $X$-$Y$ plane. Although we assume the first scenario here, an almost identical proof applies to the second scenario.}}}. 

We first give a concrete protocol for choosing the required parameters for the Raussendorf, Harrington and Goyal scheme, given the desired security threshold for the verification, see Protocol \ref{prot:Mchoice}. This will fix the size of the dotted-graph state, $N$, required for the actual computation. However as stated above, we will start with a dotted-complete graph state of size $3N$ and will break it into three smaller dotted-complete graph states of size $N$ each, see Figure \ref{fig:prot}. We will refer to these graphs as the \emph{white trap graph}, the \emph{black trap graph} and the \emph{computation graph}. In the white trap graph all the vertices in $P(\tilde{K}_N)$ will become isolated traps (called \emph{white traps}) by choosing all the vertices in $A(\tilde{K}_N)$ to be dummy qubits. Similarly in the black trap graph all the vertices in $A(\tilde{K}_N)$ will become isolated traps (called \emph{black traps}) by choosing all the vertices in $P(\tilde{K}_N)$ to be dummy qubits. We have to choose both type of vertices ($A(\tilde{K}_{3N})$ and $P(\tilde{K}_{3N})$) to be potentially isolated traps otherwise Bob could choose to cheat on one type rather than the other one. In order to make the position of traps random, Alice will choose a random partition of $P(\tilde{K}_{3N})$ into three equal size sets, and will choose appropriate dummy qubits (similar to Lemma \ref{lem:partition}) to obtain the three disconnected graphs. Note that this will lead to random positions for trap qubits, however the positions of trap qubits will be also correlated with each other and we will take care of this issue when we present the proof of the verification. The above procedure is formalized in Protocol \ref{prot:Mchoice} and finally Protocol \ref{prot:AUBQC} presents a hiding protocol that is universal, verifiable and blind. 

As a high level overview of the fault-tolerance scheme, qubits are encoded topologically as chains of defects (qubits to be measured in the $Z$ basis) of finite thickness and separation (referred to as the scale parameter) which trace out a path through the three dimensional structure of the resource state. The encoding forces non-detectable errors to be topologically non-trivial chains, either connecting or encircling defect chains. Certain Clifford group operations are implemented directly by braiding these defect chains. For the remaining operations required for universality it is necessary to implement the gate by first distilling a suitable resource state which is then used to implement the gate via teleportation (all within the topologically encoded computation). While the teleportation can be done with Clifford group operations, the distillation is implemented on a concatenated encoding where at each level of concatenation the corresponding distillation step is topologically encoded with progressively higher defect thicknesses and scale parameters. At the lowest level, however, the operations are performed directly on physical qubits, and so the defect chains are only a single qubit in diameter.
 
\begin{algorithm}
\caption{Measurement Pattern Choice}
\label{prot:Mchoice}
In what follows choosing a measurement pattern means fixing the underlying graph state together with the appropriate angles of computation such that the resulting pattern implements the desired computation due to universality. Similarly choosing a partial measurement pattern means fixing the underlying graph state together with a partial set of angles of computation corresponding to a partial computation, where the rest of angles  will be fixed in Protocol 8 where this protocol is called as a subroutine. Here, we assume that a standard labeling of the vertices of each dotted-complete graph state is known to both Alice and Bob.
\begin{enumerate}
\item Alice chooses security parameter $d$, then transforms the quantum circuit $\mathcal{C}$ corresponding to her desired computation into (or directly designs) a measurement pattern $\mathbb{M}_{Comp}$ on a graph state $\mathcal{G_L}$ which implements her computation using the encoding for topological fault-tolerant measurement-based quantum computation due to Raussendorf, Harrington and Goyal \cite{RHG07}, where $\mathcal{G_L}$ is taken to correspond to the graph state of the 3D lattice $\mathcal{L}$ introduced in \cite{RHG07} with sufficient dimensions $D_x$, $D_y$ and $D_z$ to implement her computation using an encoding with parameters as follows:
\begin{itemize}
\item Defect thickness $d$
\item Lattice scale parameter $\lambda = 5 d$
\item Distillation of resource states $\ket{A}$ and $\ket{Y}$ using $L = \lceil \log_3 (d) \rceil$ levels
\item For each concatenation level $1<\ell<L$ the thickness parameter and scale parameter for that level are chosen as $d_\ell = 3 d_{\ell-1}$ and $\lambda_\ell=\lambda_{\ell-1}$, with $d_1= 1$, $\lambda_1 = 5$, $d_L=d$ and $\lambda_L=\lambda$.
\end{itemize}
\item Alice chooses a partial measurement pattern $\mathbb{M}_{Reduce}$ which reduces the graph state $\tilde{\mathcal K}_N$ to the graph state $\mathcal{G_L}$ through Pauli measurements (Theorem \ref{thm:universal2}), where $N$ is the total number of qubits in $\mathcal{L}$.
\item Alice chooses a partial measurement pattern $\mathbb{M}_{P}$ on the graph state $\tilde{\mathcal K}_{N}$ such that every qubit corresponding to a vertex in $A(\tilde{K}_{N})$ are set to be dummy qubits. Hence all vertices in $P(\tilde{K}_{N})$ are isolated traps.
\item Alice chooses a partial measurement pattern $\mathbb{M}_{A}$ on the graph state $\tilde{\mathcal K}_{N}$ such that every qubit corresponding to a vertex in $P(\tilde{K}_{N})$ is set to be dummy qubits. Hence all vertices in $A(\tilde{K}_{N})$ are isolated traps.
\item For the graph $\tilde{K}_{3N}$, Alice chooses uniformly at random a partitioning $\mathbb{P}$ of the vertices into three equal sized sets of vertices $P_1$, $P_2$ and $P_3$.
\item Alice takes $\mathbb{M}_\mathbb{P}$ to be the partial measurement pattern where the required vertices in $A(\tilde{K}_{3N})$ are set to be dummy qubits such that the resulting state is the tensor product of three graph states of the three disconnected graphs $\tilde{k}_1= \tilde{K}_{N}$, $\tilde{k}_2= \tilde{K}_{N}$ and $\tilde{k}_3= \tilde{K}_{N}$, such that $P(\tilde{k}_i) =P_i$. 
\item Alice calculates $\mathbb{M}$, her overall measurement pattern on a graph state corresponding to $\tilde{K}_{3N}$ by combining the partial pattern $\mathbb{M}_\mathbb{P}$ with $\mathbb{M}_{Comp}$ and $\mathbb{M}_{Reduce}$ applied to subgraph $\tilde{k}_1$ and $\mathbb{M}_{P}$ and $\mathbb{M}_{A}$ applied to subgraphs $\tilde{k}_2$ and $\tilde{k}_3$ respectively, to obtain a full measurement pattern.
\end{enumerate}
\end{algorithm}

\begin{algorithm}
\caption{Verifiable Universal Blind Quantum Computation}
 \label{prot:AUBQC}
\begin{itemize}

\item \textbf{Alice's resources} \\
-- Alice chooses the pattern $\mathbb{M}$ and random partitioning $\mathcal{P}$ according to Protocol \ref{prot:Mchoice}.\\
\noindent -- The dummy qubits position, set $D$ chosen according to Protocol \ref{prot:Mchoice}.\\
\noindent -- A sequence of measurement angles, $\phi=(\phi_i)_{1\leq i \leq 3N(3N+1)/2}$ with $\phi_i \in A$, according to the description of Protocol \ref{prot:Mchoice}, where $\phi_i = 0$ for all the trap and dummy qubits. The ordering of the measurements on $P(\tilde{\mathcal{K}}_{3N})$ is chosen uniformly at random subject to the constraint that the partial ordering of measurements from $\mathbb{M}_{Comp}$ determined by flow is preserved. Such a random ordering is required to hide the position of the trap qubits. The qubits in $A(\tilde{\mathcal{K}}_{3N})$ are measured first in the order that the relevant edge entry appears in the adjacency matrix of $\mathcal{K}_{3N}$ once this random ordering has been taken into account. That is, the site in $A(\tilde{\mathcal{K}}_{3N})$ which is joined by edges to $i$ and $j$ in $P(\tilde{\mathcal{K}}_{3N})$, with $i<j$ in the random ordering imposed on $P(\tilde{\mathcal{K}}_{3N})$, is measured in position $3N(i-1) + j -\frac{i(i+1)}{2}$. Note that the measurement order of the vertices in $A$ should be independent of the computation (and traps), so in the above we prescribe one such suitable sequence. This is followed by the measurements of $P(\tilde{\mathcal{K}}_{3N})$ in the randomly chosen order.\\
\noindent -- $3N (3N+1)/2 $ random variables $\theta_i$ with value taken uniformly at random from $A$.\\
\noindent -- $3N (3N+1)/2$ random variables $r_i$ and $|D|$ random variable $d_i$ with values taken uniformly at random from $\{0,1\}$. \\
\noindent -- A fixed function $C(i, \phi_i, \theta_i, r_i, \mathbf{s})$ that for each non output qubit $i$ computes the angle of the measurement of qubit $i$ to be sent to Bob.

\item \textbf{Initial Step} \\
-- \textbf{Alice's move:} Alice sets all the value in $\mathbf{s}$ to be $0$ and prepares the qubits in the following form
\AR{
\forall i\in D &\;\;\;& \ket {d_i} \\ 
\forall i \not \in D &\;\;\;& \prod_{j\in N_G(i) \cap D} Z^{d_j}\ket {+_{\theta_i}}
}
and sends Bob all the $3N (3N+1)/2$ qubits in the order of the labeling of the vertices of the graph.

-- \textbf{Bob's move:} Bob receives $3N (3N+1)/2$ single qubits and entangles them according to $\tilde{K}_{3N}$.

\item \textbf{Step $i: \; 1 \leq i \leq 3N (3N+1)/2$}

-- \textbf{Alice's move:} Alice computes the angle $\delta_i=C(i, \phi_i, \theta_i, r_i, \mathbf{s})$ and sends it to Bob.\\ 
-- \textbf{Bob's move:} Bob measures qubit $i$ with angle $\delta_i$ and sends Alice the result $b_i$. \\
-- \textbf{Alice's move:} Alice sets the value of $s_i$ in $\mathbf{s}$ to be $s_i+r_i$. \\

\item  \label{step:Alice-prep} \textbf{Verification} \\
Alice accepts if $s_i = r_i$ for all the white and black trap qubits $i$.

\end{itemize}
\end{algorithm}

\begin{figure}[h!]
\includegraphics[width=\columnwidth]{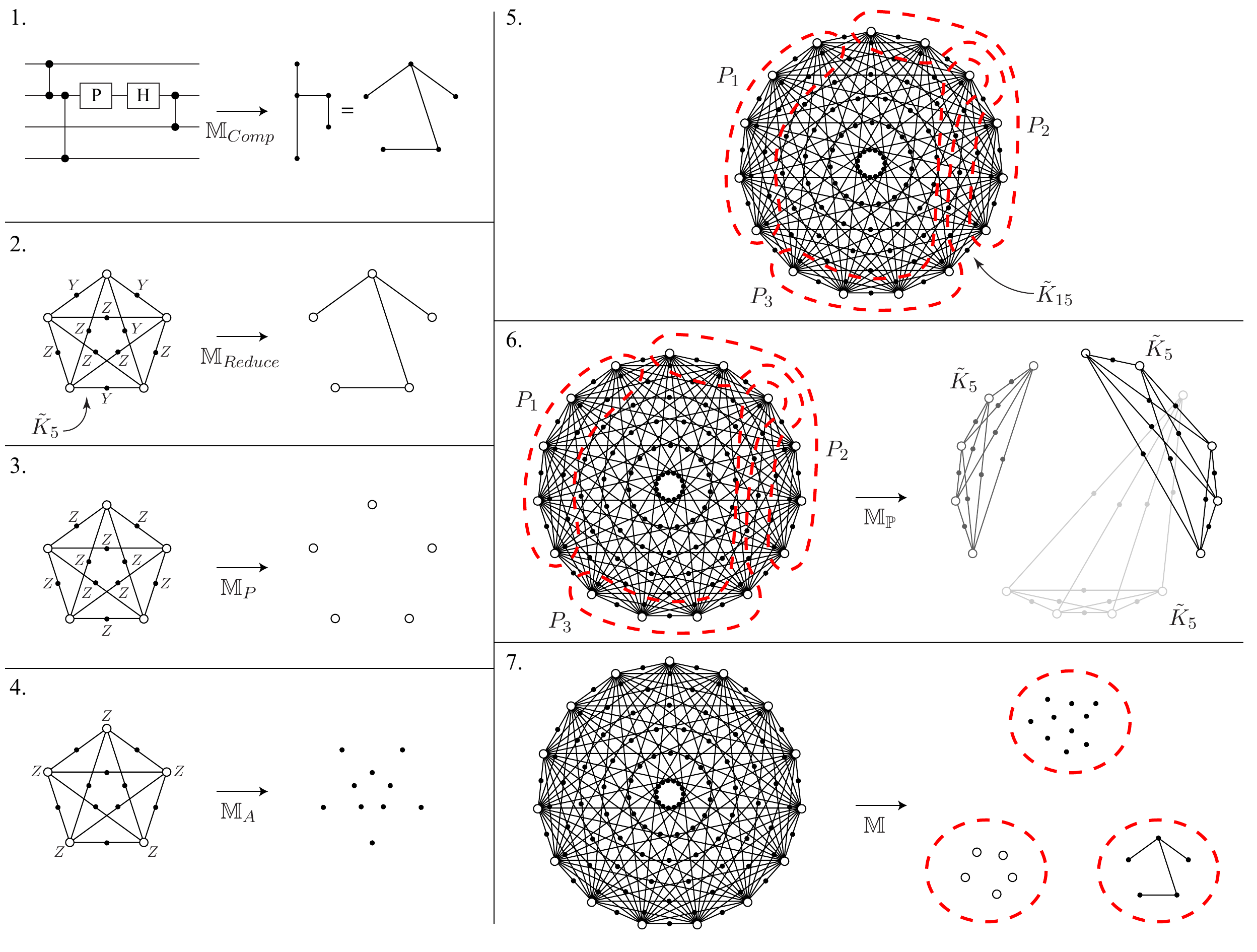}
\caption{A graphical depiction of Protocol \ref{prot:AUBQC}. In this figure we replace the Raussendorf-Harrington-Goyal encoding in the first step with a simpler computation, as to include a full encoding yields graphs too large to reasonably draw. \label{fig:prot}}
\end{figure}

\begin{theorem}
\label{th:correctness}
Assume Alice and Bob follow the steps of Protocol~\ref{prot:AUBQC}, then Alice always accepts the output and the outcome density operator is correct.
\end{theorem}
\begin{proof}
First we note that it is always possible to choose measurement patterns $\mathbb{M}_\mathcal{P}$ by Lemma \ref{lem:partition} and $\mathbb{M}_{Reduce}$ by Lemma \ref{lem:m-universal}. Further, by the universality of the Raussendorf-Harrington-Goyal encoding, it is always possible to choose $\mathbb{M}_{Comp}$. As the measurements composing $\mathbb{M}_\mathcal{P}$, $\mathbb{M}_{Reduce}$, $\mathbb{M}_{P}$ and $\mathbb{M}_{A}$ are composed entirely of Pauli basis measurements, there is no partial time ordering imposed on the sequence of measurements, and so the times at which these measurements are made have no effect on the outcome of the protocol. Thus for any honest run of the protocol, the result will be the same as if the measurements from $\mathbb{M}_{\mathcal{P}}$ were made first. By construction this measurement pattern splits the graph state into three separate graph states $\tilde{\mathcal K}_N$.

The dummy qubits in $\mathbb{M}_P$ and $\mathbb{M}_A$ correspond to break operations in their respective graphs by Lemma \ref{lem:state-graph} and hence after the initial step all the trap qubits remain unentangled from the rest. Recall that for these trap qubits $\phi_i = 0$, and since the qubit is prepared in the state $\ket{+_{\theta_i}}$ and measured in basis $\{\ket{+_{\theta_{i}}}, \ket{-_{\theta_{i}}} \}$, the measurement result communicated to Alice is $s_i=r_i$ for all such qubits. Thus, Alice always accepts, satisfying the first criterion.

By definition $\mathbb{M}_{Reduce}$ transforms the graph state corresponding to $\tilde{K}_N$ to the resource state necessary to implement $\mathbb{M}_{Comp}$. Lastly, measuring according to $\mathbb{M}_{Comp}$ yields the correct output of $\mathcal{C}$ by the correctness of the Raussendorf-Harrington-Goyal protocol.
\end{proof}

\begin{theorem} 
Protocol \ref{prot:AUBQC} is blind while leaking at most $N$.
\end{theorem}
\begin{proof}
The proof is directly obtained from Theorem \ref{prot:UBQC2}.
\end{proof}

In order to prove the verification property, as stated in Theorem \ref{t-authenT}, we require that the measurement pattern is encoded in such a way that any Pauli error of weight less than $d$ will be either corrected or detected. We now show that this is true for the Raussendorf-Harrington-Goyal scheme although this is already implicit in their paper \cite{RHG07}, we make it explicit here for completeness. In what follows, we take $\mathcal{L}$ to be the 3D lattice corresponding to the resource state used in \cite{RHG07}.

\begin{lemma}\label{lem:RHG-threshold}
Let $\mathbb{M}_\mathcal{C}$ be a measurement pattern which implements a computation $\mathcal{C}$ on $\mathcal{G_L}$, the graph state corresponding to the lattice $\mathcal{L}$, using the Raussendorf-Harrington-Goyal fault tolerance scheme with the following parameters
\begin{itemize}
\item Defect thickness $d$
\item Lattice scale parameter $\lambda = 5 d$
\item Distillation of resource states $\ket{A}$ and $\ket{Y}$ using $L = \lceil \log_3 (d) \rceil$ levels
\item For each concatenation level $1<\ell<L$ the thickness parameter and scale parameter for that level are chosen as $d_\ell = 3 d_{\ell-1}$ and $\lambda_\ell=3\lambda_{\ell-1}$, with $d_1= 1$, $\lambda_1 = 5$, $d_L=d$ and $\lambda_L=\lambda$.
\end{itemize}
Take $\sigma = \{\sigma^i\}$ to be a set of Pauli operators, such that each $\sigma^i \in \{I,X,Y,Z\}$ and acts on qubit $i$. Then for any $\sigma$, if $\mathbb{M}_\mathcal{C}$ is implemented on state $\ket{G_\mathcal{L}}$, but the output of each measurement result or unmeasured qubit $i$ is modified by applying $\sigma^i$, then either the computation is correct (corresponding to a run where all $\sigma^i = I$) or an error is detected when the output is decoded, unless $|B_\mathcal{L}| + |C_\mathcal{L}| + |D_\mathcal{L}^O| \geq 2d$, where $B_\mathcal{L} = \{\gamma: \sigma^\gamma = X\}$, $C_\mathcal{L} = \{\gamma: \sigma^\gamma = Y\}$ and $D_\mathcal{L}^O =  \{\gamma: \sigma^\gamma = Z \mbox{ and } \gamma\in O\}$, and where $O$ is the set of output (unmeasured) qubits.
\end{lemma}

\begin{proof}
In the Raussendorf-Harrington-Goyal scheme, logical qubits are topologically protected against errors. The two lowest weight topological errors are error cycles around defects and error chains running between defects. As defects have thickness $d$, any cross-section forms a rectangle of dimension at least $d \times d$ and thus perimeter at least $4(d+1)$. As an error cycle must fit around the remaining defect, the minimum error cycle is at least $4d$. As the centers of defects are separated by distance $\lambda$, the minimum distance between defects is $\lambda-d$ and hence for our parameters we have $\lambda-d = 4d$. 

The only region where this topological protection breaks down is within the regions used to distill the resource states $\ket{A}$ and $\ket{Y}$. This distillation is performed using a concatenation of $L$ levels of the Reed-Muller ($\ket{A}$) or Steane ($\ket{Y}$) codes. Each level $\ell$ of distillation is topologically protected with parameters $d_\ell$ and $\lambda_\ell$. As the Reed-Muller and Steane codes are both distance 3, an error at level $\ell$ can be caused either by a topological error at that level or not less than 3 errors at the previous level. However, since at each level $\ell<L$ we have $\lambda_\ell-d_\ell = 4d_\ell$ and $d_\ell=3d_{\ell-1}$, the minimum weight $w_\ell$ to create an error at level $\ell$ is $\min(4d_\ell, 8d_{\ell-1},4d_{\ell-1} + w_{\ell-1}, 3 w_{\ell-1})$. The four terms in this last expression account, respectively, for the minimum weight errors in each of the four possible cases: 1) The error is entirely topological at level $\ell$, 2) The error is entirely topological at level $\ell-1$, 3) the error includes both topological errors at level $\ell-1$ (which in the worst case affects two qubits with a single weight $4d_\ell$ error chain) and inherited errors from level $\ell-2$, and 4) the case where all errors are inherited from level $\ell-2$. 

We then prove that $w_\ell > 2 d_\ell$ by induction, as follows. Assume that at level $i$ we have $w_i > 2 d_i$. In that case we have $w_{i+1} = \min(4d_{i+1}, 6d_{i})$, since by assumption $4d_i + w_i>6d_i$ and $3w_i>6d_i$, and clearly $8d_i > 6 d_i$. However, we have $d_{i+1} = 3 d_i$ for all levels except the top level, where $d_L\leq 3 d_{L-1}$. Thus, in general, $2 d_{i+1} \leq 6 d_i$, and hence $w_{i+1} > 2 d_{i+1}$.  At the lowest level the error distillation uses unencoded qubits measured in non-Pauli bases, and so $w_0 = 1$, so $w_1 = 3 > 2 d_1 = 2$ and thus by induction on $i$ we obtain the result that $w_L > 2d$ as required.

Note, however, that any operation on a measured qubit which is diagonal in the computational basis ($\sigma^i \in \{I,Z\}$) does not alter the computation. Hence an undetectable logical error is not created unless the total number of measured sites for which $\sigma^i \in \{X,Y\}$ plus the total number of output qubits for which $\sigma^i \in \{X, Y, Z\}$ is equal to or greater than $2d$. Thus the outcome is either correct or when decoded results in a detected error, unless $|B_\mathcal{L}| + |C_\mathcal{L}| + |D_\mathcal{L}^O| \geq 2d$.
\end{proof}

Now we link the above general property of the Raussendorf-Harrington-Goyal scheme to our specific protocol. To do so, we first introduce the notion of independently detectable errors.
\begin{definition}
Given a dotted-complete graph state $\tilde{K}_N$, a set of output qubits $O$, a measurement pattern $\mathbb{M}_\text{target}$ containing only $X$-$Y$ plane measurements and $Z$ basis measurements, and a set of single qubit Pauli operators $\sigma = \{\sigma^i\}_{i=1}^N$ with $\sigma^i \in \{I,X,Y,Z\}$ which represent errors which modify each measurement result or unmeasured output qubit $i$ by the application of $\sigma^i$, for each location $i$ we define the set $\epsilon_i = \{i\}$ for $i\in P(\tilde{K}_N)$, and $\epsilon_i = N_{\tilde{K}_N}(i)$ for $i\in A(\tilde{K}_N)$. We say that $\sigma$ contains $k$ {\em independently detectable errors} if and only if there exists a set $\mathcal{E}$ of $k$ locations such that 
\begin{itemize}
\item For all $i\in\mathcal{E}$, $\sigma^i \in\{X,Y\}$ if $i \notin O$ or else $\sigma^i \in\{X,Y,Z\}$ if $i \in O$, and
\item $\epsilon_i \cap \epsilon_j = 0$ for all pairs $i,j \in \mathcal{E}$.
\end{itemize}\end{definition}
\noindent The intuition behind this definition is that in Protocol \ref{prot:AUBQC} the qubits in $P(\tilde{K}_{3N})$ are independently randomly distributed between the two trap graphs and the computation graph, and whether or not a qubit in $A(\tilde{K}_{3N})$ coincides with a trap or not depends only on the placement of the neighboring qubits (which are both in $P(\tilde{K}_{3N})$). The first condition ensures that the error anticommutes with some possible measurement of the system, and is hence truly an error, while the second condition ensures that we are considering only qubits associated with disjoint subsets of $P(\tilde{K}_{3N})$, and hence whether or not they coincide with a trap is uncorrelated. With this definition in place, we can proceed with proving a corollary to Lemma \ref{lem:RHG-threshold} which links that result with Protocol \ref{prot:AUBQC}.

\begin{corollary} \label{cor:weight}
Let $\mathbb{M}_\mathcal{C}$ be a measurement pattern which implements a computation $\mathcal{C}$ on graph state $\mathcal{G}_\mathcal{L}$ of $N$ vertices using the Raussendorf-Harrington-Goyal scheme with parameters
\begin{itemize}
\item Defect thickness $d$
\item Lattice scale parameter $\lambda = 5 d$
\item Distillation of resource states $\ket{A}$ and $\ket{Y}$ using $L = \lceil \log_3 (d) \rceil$ levels
\item For each concatenation level $1<\ell<L$ the thickness parameter and scale parameter for that level are chosen as $d_\ell = 3 d_{\ell-1}$ and $\lambda_\ell=\lambda_{\ell-1}$, with $d_1= 1$, $\lambda_1 = 5$, $d_L=d$ and $\lambda_L=\lambda$.
\end{itemize}
Further, let $\mathbb{M}_{Reduce}$ be a partial measurement pattern consisting of Pauli $Z$ and Pauli $Y$ measurements on qubits corresponding to the vertices in $A(\tilde{K}_N)$ which reduces $\tilde{\mathcal{K}}_N$ to $\mathcal{G}_\mathcal{L}$ up to local $Z$-rotations. Let $\mathbb{M}$ be the measurement pattern for graph state $\tilde{\mathcal K}_N$ produced by applying the partial pattern $\mathbb{M}_{Reduce}$ to the qubits corresponding to vertices in $A(\tilde{K}_N)$ and $\mathbb{M}_\mathcal{C}$ (with appropriate local $Z$-rotations applied) to the qubits corresponding to vertices in $P(\tilde{K}_N)$.

Take $\sigma = \{\sigma^i\}$ to be a set of single qubit Pauli operators, such that each $\sigma^i \in \{I,X,Y,Z\}$  acts on qubit $i$. Then for any $\sigma$, if $\mathbb{M}_\mathcal{C}$ is implemented on state $\tilde{K}_N$, but the output of each measurement result or unmeasured qubit is modified by applying $\sigma^i$, then either the computation is correct (corresponding to a run where all $\sigma^i = I$) or an error is detected when the output is decoded, unless $\sigma$ contains at least $\lceil\frac{2d}{5}\rceil$ independently detectable errors.
\end{corollary}

\begin{proof}
First we note that only qubits in $P(\tilde{K}_{3N})$ are contained in $O$, since all qubits in $A(\tilde{K}_{3N})$ will be measured to make the required resource states. All measurements on qubits associated with vertices $A(\tilde{K}_N)$ are in either the $Y$ or $Z$ basis, allowing any error in the measurement outcome to be associated with an $X$ error on the underlying qubit. As the generators for the stabilizer of $\tilde{\mathcal K}_N$ are simply the operators $X_i \prod_{j\in N_{\tilde{\mathcal{K}}_{N}}(i)} Z_j$, and each vertex in $A(\tilde{K}_N)$ has only two neighbors, both of which lie in $P(\tilde{K}_N)$, an $X$ error on a qubit associated with a vertex in $A(\tilde{K}_N)$ is equivalent to a local error on each of two qubits in $P(\tilde{K}_N)$.  Thus any local Pauli operator in $\sigma^i$ associated with a vertex in $A(\tilde{K}_N)$ can be either replaced by at most two local operators acting on qubits associated with vertices in $P(\tilde{K}_N)$ without altering the outcome of the computation, or has no effect on the computation. Note that since Pauli $Z$ operators always commute with $Z$ basis measurements, and anticommute with any measurement in the $X-Y$ plane, these local operators are always Pauli operators due to the corresponding restriction on $\mathbb{M}_\text{target}$.

The only Pauli terms which can affect the outcome of the computation are those which either flip a measurement outcome ($X$ or $Y$) or those which act non-trivially upon an unmeasured qubit (as either $X$, $Y$ or $Z$). By Lemma \ref{lem:RHG-threshold}, the outcome of the computation is unaltered unless $\sigma$ produces such errors on at least $2d$ sites. To show that this implies the existence of at least $\lceil \frac{2d}{5} \rceil$ independently detectable errors we will consider the effects of errors on $A(\tilde{\mathcal{K}}_{N})$ and $P(\tilde{\mathcal{K}}_{N})$ in relation to the resource state for the Raussendorf-Harrington-Goyal scheme, $\mathcal{G}_\mathcal{L}$. Errors on $A(\tilde{\mathcal{K}}_{N})$ only occur when the qubit in question is measured in the $Y$ basis, since for $Z$ basis measurements dummy qubits are used and the outcome of Bob's measurement is ignored. Thus, as we have shown above, such errors correspond to local Pauli errors at either end of an edge in the $\mathcal{G}_\mathcal{L}$. Errors in $P(\tilde{\mathcal{K}}_{N})$, however, correspond simply to errors on single vertices in $\mathcal{G}_\mathcal{L}$. Therefore, we can consider any error introduced by $\sigma$ as corresponding to a subgraph $g_\sigma$ of $\mathcal{G}_\mathcal{L}$, where $i \in A(\tilde{\mathcal{K}}_{N})$ introduces the vertices in $N_{\tilde{\mathcal{K}}_N}(i)$ together with a connecting edge, while  $i \in P(\tilde{\mathcal{K}}_{N})$ simply introduces the vertex $i$. Such a subgraph contains all of the qubits in $\mathcal{G}_\mathcal{L}$ which can possibly be affected by local errors after the measurement of qubits according to $\mathbb{M}_{Reduce}$ are taken into account (propagating errors from  $A(\tilde{\mathcal{K}}_{N})$ to $P(\tilde{\mathcal{K}}_{N})$). 

\begin{figure}[h!]
\begin{center}
\includegraphics[width=0.4\columnwidth]{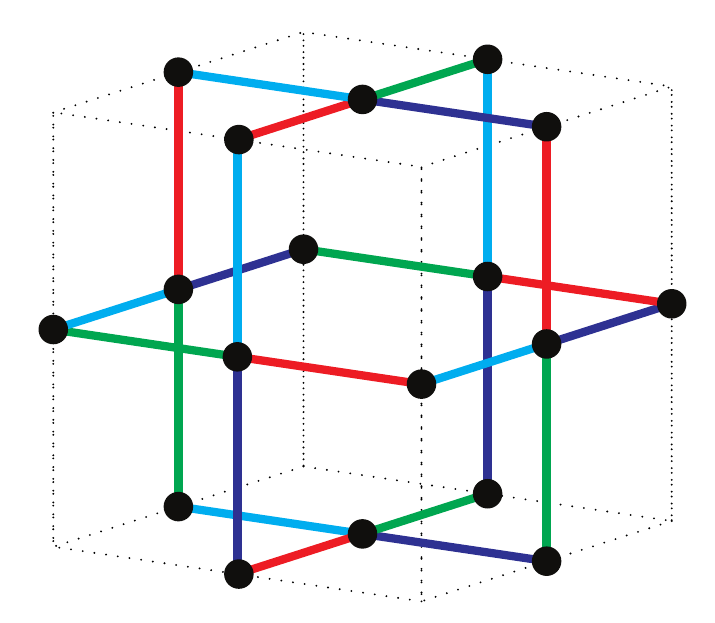}
\caption{The unit cell for the lattice corresponding to the Raussendorf-Harrington-Goyal scheme, $\mathcal{G}_\mathcal{L}$, complete with one choice of 4-edge-colouring.}
\label{fig:colouring}
\end{center}
\end{figure}

We note that any connected subgraph $g_\sigma^\gamma$ of $g_\sigma$ containing $n_\gamma$ vertices necessarily contains at least $n_\gamma-1$ edges. Note also that $\mathcal{G}_\mathcal{L}$ is 4-edge-colorable (see Figure \ref{fig:colouring}). Thus, by the pigeonhole principle, there is at least one color for that subgraph which corresponds to at least $\lceil \frac{n_\gamma-1}{4} \rceil$ edges. As the various subgraphs $g_\sigma^\gamma$ are disconnected, we are free to choose the colouring independently for each, and hence can choose a single 4-edge-colouring for $g_\sigma$ such that it includes at least $\lceil \frac{n_\gamma-1}{4} \rceil$ edges from each subgraph. We then take the set $\mathcal{E}$ to correspond to qubits in $A(\tilde{\mathcal{K}}_{N})$ corresponding to edges of this color, as well as to the single vertex in any $g_\sigma^\gamma$ for which $n_\gamma = 1$, hence $\epsilon_i \cap \epsilon_j = 0$. By Lemma \ref{lem:RHG-threshold}, this insures that the outcome of the computation is either correct or an error is detected upon decoding, or $\sigma$ contains at least $\sum_{\gamma:n_\gamma\geq 2} \lceil \frac{n_\gamma-1}{4} \rceil + \sum_{\gamma:n_\gamma=1} 1$ independently detectable errors, where $\sum_\gamma n_\gamma \geq 2d$. Note that
\begin{align*}
\sum_{\gamma:n_\gamma\geq 2} \lceil\frac{n_\gamma-1}{4}\rceil + \sum_{\gamma:n_\gamma=1} 1 \geq \frac{2d}{5},
\end{align*}
and hence  the computation is either correct or an error is detected upon decoding, or $\sigma$ contains at least $\lceil\frac{2d}{5}\rceil$ independently detectable errors.
\end{proof}

The above corollary guarantees that one of the condition of Theorem \ref{t-authenT} for the verification with the amplified security is satisfied. However we cannot yet directly use that theorem since, as stated before, the position of the traps are not completely random as the position of the black traps are fixed once we choose the random position assignment of qubits in $P(\tilde{\mathcal{K}}_{3N})$ to each of the three subgraphs. This is why we have introduced the notion of independently detectable errors. Here we give a direct proof of verification for Protocol \ref{prot:AUBQC} following the same steps as the proof of Theorem \ref{t-authenT}.

\begin{theorem}
\label{thm:verification}
Protocol \ref{prot:AUBQC} is in general $(5/6)^{\lceil\frac{2d}{5}\rceil}$-verifiable, and in the case of only classical output is $(2/3)^{\lceil\frac{2d}{5}\rceil}$-verifiable, where $d$ is the security parameter as described in Protocol \ref{prot:Mchoice}.
\end{theorem}

\begin{proof}

The proof of this theorem follows the same strategy as Theorem \ref{t-authen1}, first taking the most general strategy for Bob, expanding this in terms of Pauli operators, and lastly showing that any Pauli term which leads to an incorrect outcome is detected with high-probability. We note that any deviation by Bob from Protocol \ref{prot:AUBQC} can be rewritten in the form shown in Figure \ref{fig:devfig2}. The proof of this is identical to the corresponding step in the proof of Theorem \ref{t-authen1}: Without loss of generality any deviation by Bob from the protocol can be written in the form of Figure \ref{fig:devfig1}. We can treat $\{\delta_i\}$ as inputs to the circuit without violating causality, as they do not interact with any other part of the computation until after $b_j$ has been measured, for all $j<i$. Then simply by reordering the operators via their commutation relations we obtain the form in Figure \ref{fig:devfig2} as required. As a result, any deviation by Bob can be written as a single deviation operator $\Omega$ which acts upon the quantum states Bob receives from Alice as well as $\delta_i$ and some private register held by Bob. Similar to the proof of Theorem \ref{t-authen1} the probability of Alice accepting an incorrect outcome density operator is then
\begin{align*}
p_\text{incorrect} &= \sum_\nu p(\nu)\mbox{Tr}\left( P^\nu_\text{incorrect} B_j(\nu)\right)\\
&= \sum_{b,\nu} p(\nu) \mbox{Tr} \bigg(\displaystyle P_\text{incorrect}\ket{b+c_r}\bra{b} C_{\nu_C,b}\Omega \left( \mathcal {P}  | \Psi^{\nu,b}\rangle \langle\Psi^{\nu,b} | \mathcal {P}^\dagger \right)C_{\nu_C,b}^\dagger \ket{b}\bra{b+c_r} \bigg)\\
&= \sum_{k,b,i,j,\nu} p(\nu) \alpha_{ki} \alpha_{kj}^* \mbox{Tr} \bigg(\displaystyle P_{\bot} \left(\bigotimes_{t\in T} \ket {\eta_t^{\nu_T}} \bra {\eta_t^{\nu_T}} \right) \ket{b+c_r}\bra{b} \\
&~~~~~~~~~~~~~~~~~~~~~~~~~~~~~~~~~~~~~C_{\nu_C,b}\sigma_i \mathcal {P}  | \Psi^{\nu,b}\rangle \langle\Psi^{\nu,b} | \mathcal {P}^\dagger \sigma_j C_{\nu_C,b}^\dagger \ket{b}\bra{b+c_r} \bigg),
\end{align*}
where as in previous proofs, we take the Kraus operators associated with the $\Omega$, once Bob's private system has been removed, to be $\chi_k = \sum_i \alpha_{ki} \sigma_i$, with $\sum_k \sum_i \alpha_{ki} \alpha_{ki}^* = 1$. 

By Corollary \ref{cor:weight}, $P_\bot$ projects out the terms in the above sum where $\sigma_i$ does not contain at least $\lceil\frac{2d}{5}\rceil$ independently detectable errors on the computation graph. This is a somewhat stronger condition than we actually need, and so we will consider terms corresponding to any $\sigma_i$ which produces at least $\lceil\frac{2d}{5}\rceil$ independently detectable errors in total across all three subgraphs (the computation graph and the two trap graphs). We will denote by $\mathcal{I}$ the set of all $i$ for which $\sigma_i$ does not satisfy this condition. Similar to the proof of Theorem \ref{t-authen1}, all terms for which $i\neq j$ average to zero. Thus, as in the proof of Theorem \ref{t-authenT}, we obtain
\begin{align*}
p_\text{incorrect} \leq {\sum_{k}} \sum_{i\notin \mathcal{I}} \sum_T p(T) |\alpha_{ki}|^2 \prod_{t\in T} \left( \sum_{\theta_t,r_t} p(\theta_t)p(r_t) \langle\eta_t^{\nu_T}|\sigma_i |\eta_t^{\nu_T}\rangle^2 \right).
\end{align*}
As before, we introduce notional sets $S_\gamma$ of three qubits each such that exactly one qubit from each set is on each of the three subgraphs (the two trap graphs and the computation graph), and where either all of the qubits are in $P(\tilde{\mathcal{K}}_{3N})$ or all of the qubits are in $A(\tilde{\mathcal{K}}_{3N})$ (ensuring exactly one trap and at least one dummy qubit per set). As every $\sigma_i$ in the above sum corresponds to at least $\lceil \frac{2d}{5}\rceil$ independently detectable (and hence uncorrelated) errors across these sets $S_\gamma$, we have 
\begin{align*}
p_\text{incorrect} &\leq {\sum_{k}} \sum_{i\notin \mathcal{I}} |\alpha_{ki}|^2 \prod_{\gamma} \left( \sum_{t_\gamma,r_{t_\gamma},\theta_{t_\gamma}} p(t_\gamma)p(r_{t_\gamma})p(\theta_{t_\gamma})\langle\eta_t^{\nu_T}|\sigma_i |\eta_t^{\nu_T}\rangle^2 \right)\\
&=  {\sum_k}\sum_{i\notin \mathcal{I}} |\alpha_{ki}|^2 \prod_{\gamma}\left(\sum_{t_\gamma,r_{t_\gamma},\theta_{t_\gamma}}\frac{1}{48}\langle\eta_t^{\nu_T}|\sigma_i |\eta_t^{\nu_T}\rangle^2\right),
\end{align*}
where as before $t_\gamma$ denotes the location of the trap qubit in set $S_\gamma$. Averaging over all values of $t_\gamma$, $r_{t_\gamma}$ and $\theta_{t_\gamma}$, we obtain
\begin{align*}
p_\text{incorrect} &\leq  {\sum_k} \sum_{i\notin \mathcal{I}} |\alpha_{ki}|^2 \prod_{\gamma} \left(1 - \frac{w_\gamma}{6}\right)\\
&\leq  {\sum_k} \sum_{i\notin \mathcal{I}} |\alpha_{ki}|^2 \prod_{\gamma} \left(1 - \frac{1}{6}\right)^{w_\gamma}\\
&=  {\sum_k} \sum_{i\notin \mathcal{I}} |\alpha_{ki}|^2 \left(\frac{5}{6}\right)^{\sum_\gamma w_\gamma}\\
&\leq {\sum_k} \sum_{i\notin \mathcal{I}} |\alpha_{ki}|^2 \left(\frac{5}{6}\right)^{\lceil \frac{2d}{5}\rceil}\\
&\leq \left(\frac{5}{6}\right)^{\lceil \frac{2d}{5}\rceil},
\end{align*}
where $w_\gamma$ denotes the number of independently detectable errors which fall within set $S_\gamma$. In the special case of all classical output, however, the bound can be made tighter, since $|\eta^\nu_{t_\gamma}\rangle=|r^\nu_{t_\gamma}\rangle$, and hence
\begin{align*}
p_\text{incorrect} &\leq {\sum_k} \sum_{i\notin \mathcal{I}} |\alpha_{ki}|^2 \prod_{\gamma} \mbox{Tr}\left(\sum_{t_\gamma,r_{t_\gamma}}\frac{1}{6}\langle r^\nu_{t_\gamma}|\sigma_{i|t}|r^\nu_{t_\gamma}\rangle^2\right)\\
&\leq {\sum_k} \sum_{i\notin \mathcal{I}} |\alpha_{ki}|^2 \prod_{\gamma} \left(1 - \frac{w_\gamma}{3}\right)\\
&\leq {\sum_k}  \sum_{i\notin \mathcal{I}} |\alpha_{ki}|^2 \prod_{\gamma} \left(1 - \frac{1}{3}\right)^{w_\gamma}\\
&=  {\sum_k} \sum_{i\notin \mathcal{I}} |\alpha_{ki}|^2 \left(\frac{2}{3}\right)^{\sum_\gamma w_\gamma}\\
&\leq {\sum_k}  \sum_{i\notin \mathcal{I}} |\alpha_{ki}|^2 \left(\frac{2}{3}\right)^{\lceil \frac{2d}{5}\rceil}\\
&\leq \left(\frac{2}{3}\right)^{\lceil \frac{2d}{5}\rceil}.
\end{align*}\end{proof}

\section{Conclusions and discussion}

We have extended the original universal blind quantum computing (UBQC) protocol presented in \cite{BFK09} with new concepts of blind preparation of isolated dummy qubits (a qubit prepared randomly in the set $\{\ket 0 , \ket 1\}$) and isolated trap qubits (a qubit prepared randomly in the set $\{\ket +_\theta\}$). These two modifications lead to a new construction for unconditionally verifiable blind quantum computation. However, in this way only polynomially bounded security could be achieved. Building upon these ideas, combined with fault-tolerant computation, we presented a new UBQC protocol that achieve exponentially bounded security for the verification scheme using new resource state, the dotted-complete graph state. The new protocol extend the topological fault-tolerant measurement-based quantum computation scheme due to Raussendorf, Harrington and Goyal \cite{RHG07} to a blind setting. We note that while consideration of fault-tolerance in the blind computation itself is beyond the scope of the present work, if Protocol \ref{prot:AUBQC} is modified so as to allow Alice to accept a finite error rate on the trap qubits, the probability of Bob successfully cheating is exponentially suppressed in the gap between the expected error weight inferred from trap measurements and our threshold of $\lceil \frac{2d}{5} \rceil$, and so a fault-tolerant adaptation of this protocol should be possible.

As mentioned before, a verifiable UBQC protocol can be viewed as an interactive proof system where Alice acts as the verifier and Bob as the prover \cite{Dorit,BFK09,RUV13}. This link to complexity theory suggests a novel approach to questions such as the open problem of finding an interactive proof for any problem in BQP with a BQP prover, but with a purely classical verifier. The conceptual link between blindness and interactive proof systems is the key ingredient for verifying the ``high complexity'' quantum-theoretical models with ``low complexity'' classical ones.

\section*{Acknowledgements}

We thank Anne Broadbent for endlessly many insightful discussions throughout the writing of this paper. We would also like to acknowledge Robert Raussendorf and Earl Campbell for their help on the properties of the topological fault tolerance scheme. We also thank Vedran Dunjko, Iordanis Kerenedis, Urmila Mahadev and Tomoyuki Morimae for helpful discussions on the proof of Theorem \ref{t-authen1} and for pointing out to us an error in the first draft. JF acknowledges support from the National Research Foundation and Ministery of Education, Singapore. This material is based on research supported
in part by the Singapore National Research Foundation under NRF Award No. NRF-NRFF2013-01. EK acknowledges support from Engineering and Physical Sciences Research Council grant EP/E059600/1.

\bibliographystyle{siam}
\bibliography{BlindQC}

\end{document}